\documentclass[onecolumn,floats,floatfix,showpacs,nofootinbib, 12 pt]{revtex4-1}

\usepackage{graphicx}
\usepackage{dcolumn}

\usepackage{bm}
\usepackage{graphics}
\usepackage{slashed}
\usepackage{amssymb}
\usepackage{natbib}
\usepackage{amsmath}
\usepackage{url}
\usepackage[usenames,dvipsnames,svgnames,table]{xcolor}
\usepackage{amsfonts}
\usepackage{subfig}
\usepackage{amsthm}
\usepackage[breaklinks=true]{hyperref}

\usepackage{color}
\newcommand{\bea}{\begin{eqnarray}}
\newcommand{\ena}{\end{eqnarray}}
\newcommand{\bean}{\begin{eqnarray*}}
\newcommand{\enan}{\end{eqnarray*}}

\newtheorem{theorem}{Theorem}
\newtheorem{lemma}{Lemma}

\newtheorem{definition}{Definition}

\input{tcilatex}

\begin{document}

%


\title{On necessary and sufficient conditions for strong hyperbolicity in systems with constraints}

\author{Fernando Abalos${}^{1}$}\email{abalos@famaf.unc.edu.ar}

\author{Oscar Reula${}^{1}$}\email{oreula@unc.edu.ar}

\affiliation{${}^{1}$ Facultad de Matem\'atica, Astronom\'\i{}a y F\'\i{}sica, Universidad Nacional de C\'ordoba and IFEG-CONICET, Ciudad Universitaria, X5016LAE C\'ordoba, Argentina }

%
%
\begin{abstract}
In this work we study constant-coefficient first order systems of partial differential equations and give necessary and sufficient conditions for those systems to have a well posed Cauchy Problem. In many physical applications, due to the presence of constraints, the number of equations in the PDE system is larger than the number of unknowns, thus the standard Kreiss conditions can not be directly applied to check whether the system admits a well posed initial value formulation. In this work we find necessary and sufficient conditions such that there exists a reduced set of equations, of the same dimensionality as the set of unknowns, which satisfy Kreiss conditions and so are well defined and properly behaved evolution equations. We do that by studying the systems using the Kronecker decomposition of matrix pencils and, once the conditions are meet, finding specific families of reductions which render the system strongly hyperbolic.
We show the power of the theory in some examples: Klein Gordon, the ADM, and the BSSN equations by writing them as first order systems, and studying their Kronecker decomposition and general reductions.

\end{abstract}

\maketitle

\section{\label{sec:level1:1}Introduction}

%
In \cite{geroch1996partial} Geroch introduces a general setting for dealing
with first order systems of partial differential equations. The novelty of
his approach was that by keeping the description covariant, that is without
choosing an evolution time nor a time-space splitting, several features of
the underlying structure of these systems became apparent: First, there is a
notion of constraint equations which is well defined and does not depend on
the introduction of any preferred hyper-surface, and second there is in
general non-natural notion of an "evolution system". Constraints are certain
linear combinations of the equations in the system that satisfy some
property, while evolution equations are another linear combinations which we
shall call \textbf{reductions}; when these reductions give rise to a well
posed set of evolution equations we call them \textbf{hyperbolizers}. Well
posedness, the assertion that solutions depend continuously on their initial
data, is a necessary condition on any physical theory to have
predictability. Well posedness in particular becomes crucial when trying to
find numerical solutions, see for instance \cite{0264-9381-18-17-202}. In
this work it will be necessary to enlarge the class of allowed reductions,
they would not just be multiplicative linear combinations, but we shall also
allow pseudo-differential ones (keeping their degree to zero). It is in this
extended class that we can find necessary and sufficient conditions for the
existence of hyperbolizers. This extension arises naturally, and an
extensive literature about the theory of pseudo-differential operators can
be consulted \cite{taylor1991pseudodifferential, taylor1996pseudodifferential, nagy2004strongly, schulze1991pseudo, taylor2013partial, calderon1962existence, hormander1966pseudo, hadamard2014lectures, petersen1983introduction, nirenberg1973lectures, friedrichs1970pseudo, treves1980introduction, lax1963l_2, kohn1973pseudo, hormander1965pseudo}, etc.

The problems in the cases with no constraints present (a well defined
statement), and where the system is consistent, that is, the number of
equations coincides with the number of unknowns, have been fully understood in the
celebrated Kreiss Matrix Theorem \cite{kreiss1962stabilitatsdefinition, gustafsson1995time, kreiss2004initial}. 
This theorem does so by stating several equivalent conditions for the system to admit a
hyperbolizer, which in the general cases is some pseudo-differential
operator. Once one of these hyperbolizers is found, it is used for the
construction of energy estimates, which in turn are used for establishing
well posedness (\cite{sarbach2012continuum, taylor1996pseudodifferential, metivier20142}). In \cite%
{abalos2017necessary}. Recently a new necessary condition was found for these
particular class of systems, namely systems without constraints. It involves
the use of the Singular Value Decomposition (SVD) of the principal symbol of
the system. The strength of this new condition manifests itself in the fact
that, contrary to the others in Kreiss theorem,  it can be applied to generic first order systems,
that is, not necessarily squared ones. 
If that condition were to fail for a first order set of equations, 
then there would be no reduction which would make it strongly hyperbolic. 
Thus, a powerful tool has been developed to easily rule out theories which fail it. 
In this work we refine the condition in order for it to also become a sufficient condition.

The theory we shall develop can be quite general, but in this article we
shall restrict to linear, constant coefficient systems. This restriction
would allow us to make simpler assertions and, correspondingly, simpler
proofs. The general theory will be spelled out in a more technical paper.
Nevertheless most of the material here introduced applies to generic first
order quasi-linear systems\footnote{ Another approach to second order partial differential equations is studied in \cite{Gundlach_2006}.}.

Our approach consists in: first, choosing a hypersurface, at each point of
it, the co-vector normal to it allows us to transform the principal symbol
into a matrix pencil, second applying the Kronecker decomposition to it,
thus obtaining the intrinsic structure of the differential equations and
finally building an specific reduction of the system. The Kronecker
decomposition allows us to recognize in its blocks the evolution of the
physically relevant fields and constraint parts of the differential
equations. The blocks related to constraint propagation admit many different
reductions, in particular, it is possible to build reductions with  arbitrary 
finite constraint propagation speeds. A similar technique for the case in
which the space-time is two dimensional was used by Pavel Motloch and et.
al. \citep{PhysRevD.93.104026}.

In section \ref{Strongly} we introduce Geroch's formalism so as to fix
notation. We then define, in this covariant setting strong hyperbolicity,
and so introduce the hyperbolizers. The definitions we introduce are such
that these reductions, in more than $1+1$ dimensions, can be
pseudo-differential, namely they can depend not only on the sections of the
bundle, but also in the co-tangent bundle of the base manifold. For generic
quasi-linear systems well posedness, as we understand it today, needs as a
sufficient condition smoothness with respect to this co-tangent bundle.
Since in this work the theory is restricted to linear constant-coefficient
systems, only an algebraic condition suffices, that is, no smoothness
condition is needed. We finally state the main theorem of the theory. We do
it in steps, first we state a theorem asserting the equivalence of our new
conditions to those of Kreiss matrix theorem. A new feature of this new
condition is that we only need to look at certain matrix pencils in a
neighborhood of their generalized eigenvalues. With this tool at hand we can
easily state our main theorem for generic systems.

In section \ref{Proof_theorems} we built the necessary ingredients for
proving our main theorem. Essentially we look for a Kronecker decomposition
of the principal symbol at certain points of the characteristic surfaces and
show that the condition hyperbolicity limits the possible Kronecker blocks
to only two types. Of those two blocks, one are the Jordan blocks. Under the
condition on the angles stated in the main theorem these Jordan blocks can
only be diagonal. Once this is established a hyperbolizer can be easily
constructed. We still need to prove the uniformity of our construction. For
that we use the same condition, which is an uniformity condition on the
angles that two kernel subspaces form between each other, to infer the
uniformity needed for strong hyperbolicity. In practical applications, for a
given system, there is a simple algorithm to compute such angles, so this
condition is really helpful in understanding possible new theories. The
general theory provided by the Kronecker decomposition allows more general
types of constraints than those appearing in Geroch's formalism. They
essentially reflect the existence of constraints which contain higher order
derivatives of the fields. They appear as higher dimensional blocks in that
decomposition. Nevertheless we have found that all of them can be readily
taken care of with an appropriate reduction. Unfortunately we do not have
any physically relevant example of these types of constraints. Nature seems
to prefer the lowest order ones.

Finally in section \ref{sec:level1:4} we introduce some examples, the Klein
Gordon, the ADM and the BSSN equations, which illustrates the power of the theory. We finish the
work with several appendices where, besides proving the new Kreiss
condition, we have included preliminary material and notation. 


\section{\label{Strongly}The setting and the Main Theorem}


We consider constant coefficients first order systems of the form 
\begin{equation}
\mathfrak{N}_{~\alpha}^{Aa}\nabla_{a}\phi^{\alpha}=0  \label{sht_1}
\end{equation}
over a real manifold $M$, with $x^{a}$ point of $M$ and $\dim M=n+1$. We
follow the notation of \cite{geroch1996partial} and \cite%
{abalos2017necessary}. Here the fields $\phi^{\alpha}$ are the unknown
fields and $\mathfrak{N}_{~\alpha}^{Aa}$ is a given constant tensor field
that depends on the particular physical theory under study. They are
sections on a bundle with a vector fiber which we shall denote by $\Phi_R$.
Lower letters $a,b,c$ represent space-time indices, Greek indices $%
\alpha,\beta,\gamma$ represent field indices $\left\vert \alpha\right\vert
:=\dim\left( "\alpha"\right) =u,$ and capital letters A, B,.. represent
multi-tensorial indices on the fiber space of equations $\left\vert
A\right\vert :=\dim\left( "A"\right) =e$. We shall denote its vector space
by $\Psi_{L}$ and consider systems that have at least the same number of
equations than fields, so $e\geq u$.

We are only interested in strongly hyperbolic
systems. Since they are stable under lower order terms additions, 
in the analysis covariant derivatives can be
exchanged for partial derivatives. For the same reason we set any lower
order term to zero. 

In our description we shall introduce a particular local co-vector field, $%
n_{a}$, it is then convenient to adapt a coordinate system to it in such a
way that $n_{a}=\nabla _{a}t$ where $t$ it is called the time coordinate
and it is a function that its level surfaces define a local foliation of $M$
by hyper-surfaces $\Sigma _{t}$. Then the set of coordinates $x^{a}=\left(
t,x^{1},...,x^{n}\right) $ define a Gaussian normal coordinates adapted to
this foliation. Consider the vector $t^{a}=\left( 1,0,0,0\right) $ such that 
$t^{a}n_{a}=1$ and $t^{a}\partial _{a}=\partial _{t}$ and the projector $%
m_{~b}^{a};=\delta _{~b}^{a}-t^{a}n_{b}$ (where $\delta _{~b}^{a}$ is the
identity map) such that $m_{~b}^{a}t^{b}=0$ and $m_{~b}^{a}n_{a}=0.$ Then
equation (\ref{sht_1}) could be written as 
\begin{equation*}
\mathfrak{N}_{~\alpha }^{Aa}n_{a}t^{b}\partial _{b}\phi ^{\alpha }+\mathfrak{%
N}_{~\alpha }^{Aa}m_{~a}^{b}\partial _{b}\phi ^{\alpha }=\mathfrak{N}%
_{~\alpha }^{Aa}n_{a}\partial _{t}\phi ^{\alpha }+\mathfrak{N}_{~\alpha
}^{Aa}m_{~a}^{b}\partial _{b}\phi ^{\alpha }=0.
\end{equation*}%
Notice that the term $m_{~a}^{b}\partial _{b}$ have no temporal partial
derivatives.

Since we are considering constant coefficient problems, we can Fourier
transform in space coordinates and reduce the system to the following
equivalent system
\begin{equation}
\mathfrak{N}_{~\alpha }^{Aa}n_{a}\partial _{t}\hat{\phi}^{\alpha }+ i \mathfrak{%
N}_{~\alpha }^{Aa}k_{a}\hat{\phi}^{\alpha }=0  \label{Fourier_2}
\end{equation}
where $k_{a}t^{a}=0$,  $k_{a}m_{~b}^{a}=k_{b}$ and with initial data over the hyper-surface $t=0$ 
\begin{equation*}
\left. \phi ^{\alpha }\right\vert _{t=0}=\hat{\phi}_{initial}^{\alpha
}e^{i\left( k_{a}\right) x^{a}}.
\end{equation*}

Since the frequency $k_{a}$ in the initial data is fixed but arbitrary, we
look for solutions of equation (\ref{Fourier_2}) for all $k_{a}$ not
proportional to $n_{a}.$

In general, equation system (\ref{Fourier_2}) has more equations than
fields, in particular there are $c$ linear combinations of equations without
time derivatives. They are called differential constraints, for a formal and
geometrical definition see \cite{geroch1996partial}. We are going to
restrict consideration to those systems where the number of equations
satisfies $e=u+c$ $\ $where $c$ is the number of constraints. In Geroch's
terminology they are called complete.

These constraints restrict the available initial data and for consistency it
must be shown that if initially satisfied they remain so along the
evolution. We shall not deal with this problem in this work, assuming this
is so, since it involves integrability conditions which depends on lower
order terms.

While in Geroch's formalism constraint equations are singled out, evolution
equations are not. They are not unique and further structure must be
introduced to single out a particular set of them. Given a particular set of
evolution equations, linear combinations of constraints can be added to
generate another equivalent system. They are not naturally unique. To single
out a particular set we introduce a new tensor field $h_{~A}^{\gamma }$ that
reduces the system to a set of purely evolution equations, 
\begin{equation}
h_{~A}^{\alpha }\mathfrak{N}_{~\alpha }^{Aa}n_{a}\partial _{t}\hat{\phi}%
^{\alpha }+ i h_{~A}^{\gamma }\mathfrak{N}_{~\alpha }^{Aa}k_{a}\hat{\phi}%
^{\alpha }=0.  \label{red_sys}
\end{equation}%
This set has $u$ independent equations, as many as there are fields. We
shall refer to $h_{~A}^{\alpha }$ as a reduction. In general it will depend
on the wave number vector $k_{a}$.

We shall call system (\ref{sht_1}) strongly hyperbolic if there is a
reduction such that system (\ref{red_sys}) is so, using the usual
definition, namely definition \ref{Definition:2} below.

We first need to introduce another definition, assuming that $h_{~A}^{\alpha }\mathfrak{N}_{~\alpha }^{Aa}n_{a}$ is
invertible (a necessary condition for hyperbolicity), we define
\begin{equation}
A_{~\gamma }^{\alpha a}k_{a}:=\left( \left( h_{~A}^{\alpha }\mathfrak{N}%
_{~\alpha }^{Aa}n_{a}\right) ^{-1}\right) _{~\gamma }^{\alpha
}h_{~A}^{\gamma }\mathfrak{N}_{~\alpha }^{Aa}k_{a}.  \label{A_eq}
\end{equation}
In addition, in the following definitions when we say "for all $k_{a}$" we
mean "all $k_{a}$ not proportional to $n_{a}$ and $\left\vert k\right\vert
=1,$ with $\left\vert \cdot \right\vert $ some positive definite norm".

\begin{definition}
\label{Def:hyp} System (\ref{red_sys}) is called hyperbolic if for all $k_{a}
$, $A_{~\gamma }^{\alpha a}k_{a}$ has only real eigenvalues.
\end{definition}

This definition means that all propagation velocities are real, so no
exponential growth with frequency can be expected, although a polynomial growth is possible~\footnote{For non-constant coefficients systems that growth can even become exponential, see for instance \cite{kreiss2004initial}}. This is not by itself
sufficient for stability and well posedness but it is certainly necessary.
Following the Kreiss's Matrix theorem \cite{kreiss1962stabilitatsdefinition}%
, \cite{kreiss2004initial} we now estate several 
necessary and sufficient conditions for strong hyperbolicity of evolution equations:

\begin{definition}
\label{Definition:2} We call system (\ref{red_sys}) strongly hyperbolic if
any of the following four equivalent conditions hold:

1- System (\ref{red_sys}) is hyperbolic and $A_{~\gamma }^{\alpha a}k_{a}$
it is uniformly diagonalizable: that is, for all $k_{a}$ there exist $%
S_{~\rho }^{\alpha }\left( k\right) $, and $C>0$ such that
\newline $A_{~\gamma
}^{\alpha a}k_{a}=S_{~\rho }^{\alpha }\left( k\right) \Lambda _{~\tau
}^{\rho }\left( k\right) \left( S^{-1}\right) _{~\gamma }^{\tau }\left(
k\right) $ with $\Lambda _{~\tau }^{\rho }\left( k\right) $ being diagonal
and \newline $\left\vert S\left( k\right) \right\vert \left\vert S^{-1}\left(
k\right) \right\vert \leq C.$

2- For all $k_{a}$ and all $s\in \mathbb{C}$ with $Im(s)>0$, there
exists a constant $C>0$ such that,

\begin{equation}
\left\vert \left( A_{~\gamma }^{\alpha a}k_{a}-s\delta _{~\gamma }^{\alpha
}\right) ^{-1}\right\vert \leq \frac{C}{Im(s)}  \label{kreiss_1}
\end{equation}

3- For all $k_{a},$ there exists an positive definite Hermitian form%
\footnote{%
The $\tbinom{0}{2}$ tensor $H_{\alpha \eta }$ is a Hermitian form if $%
H_{\alpha \eta }=\overline{H_{\eta \alpha }}$.
\par
And It is positive definite if $\overline{u^{\alpha }}H_{\alpha \eta
}u^{\eta }>0$ for all $u^{\alpha }\in T_{x}M$ and for all $x\in M.$} $%
H(k)_{\alpha \beta }$ and a constant $C>0$ such that,

$i)$ $H(k)_{\alpha \eta }A_{~\gamma }^{\alpha a}k_{a}\;\;$is an Hermitian
form, i.e. $H(k)_{\alpha \eta }A_{~\gamma }^{\alpha a}k_{a}=H(k)_{\gamma
\eta }\overline{A_{~\alpha }^{\eta a}}k_{a}$

$ii)\frac{1}{C}$ $H_{\delta \gamma }^{0}\geq H(k)_{\alpha \delta }\geq
CH_{\delta \gamma }^{0}>0\;\;\;\forall k_{a},$

where $H_{\alpha \gamma }^{0}$ is a positive definite Hermitian form that
does not depends on $k_{a}$.

4- For all $k_{a}$ and $%
t\geq 0$, there exists  $C>0$ such that,$%
\left\vert e^{itA_{~\gamma }^{\alpha a}k_{a}}\right\vert \leq C$.
\end{definition}

For real equation systems, Hermiticity has to be understood by symmetry in
the corresponding indexes.

The question is then: \textit{Under which circumstances do there exist
reductions which make the system (\ref{red_sys}) strongly hyperbolic?}.
Clearly the conditions for the existence of such reductions, $h_{~A}^{\gamma
}$, which we shall call from now on hyperbolizers, depends only on the
properties of the principal symbol, in particular on the behavior of $%
\mathfrak{N}_{~\alpha }^{Aa}l\left( \lambda \right) _{a}$ along the set of planes $S_{n_{a}}^{%
\mathbb{C}}=\{l\left( \lambda \right) _{a}:=-\lambda n_{a}+k_{a}\}$, for all 
$k_{a}$ not proportional to $n_{a}$, with $\left\vert k\right\vert =1$ and $%
\lambda \in \mathbb{C}$. More specifically, we shall concentrate on
neighborhoods of real lines on these planes. The real lines given by, $S_{n_{a}}
$ with $\lambda \in \mathbb{R}.$ The condition $k_{a}$ not proportional to 
$n_{a}$ implies that these planes and lines do not cross the origin for any 
$\lambda$. Each complex plane depends on some $k_{a}$ but we shall call them
generically $l_{a}(\lambda )$ in order not to obfuscate the notation.

Notice that if we propose a plane wave solution, $\hat{\phi}^{\alpha
}=\delta \phi ^{\alpha }e^{i\left( -\lambda n_{a}\right) x^{a}}$ for 
(\ref{red_sys}), we arrive to an equation for the right kernel of the principal
symbol, 
\begin{equation}
h_{A}^{\gamma }\mathfrak{N}_{~\alpha }^{Ab}l\left( \lambda \right) _{b}\text{
}\delta \phi ^{\alpha }=\left[ \lambda \left( -h_{A}^{\gamma }\mathfrak{N}%
_{~\alpha }^{Ab}n_{b}\right) +\left( h_{A}^{\gamma }\mathfrak{N}_{~\alpha
}^{Ab}k_{b}\right) \right] \delta \phi ^{\alpha }=0  \label{symb_1}
\end{equation}%
where $l\left( \lambda \right) _{b}=-\lambda n_{a}+k_{a}\in S_{n_{a}}^{%
\mathbb{C}}$. Here the unknown are $\delta \phi ^{\alpha }$ and $\lambda $.
The completeness of these plane wave solutions is key to understand well
posedness. Thus, we shall next study the kernels of the principal symbol $%
\mathfrak{N}_{~\eta }^{Bb}l_{b}$ along $S_{n_{a}}^{\mathbb{C}}$. We shall
call right kernel to the subspace of vectors $\delta \phi ^{\eta }$ such
that $\left( \mathfrak{N}_{~\eta }^{Bb}l_{b}\right) \delta \phi ^{\eta }=0$
and left kernels to the subspace of co-vectors $X_{B}$ such that $%
X_{B}\left( \mathfrak{N}_{~\eta }^{Bb}l_{b}\right) =0$ \footnote{%
Right kernel will be vectors that contract with down indices on the
operator, and left kernel will be co-vectors that contract to up indices.}.

In analogy to the our first definition \ref{Def:hyp}, we define
hyperbolicity for the whole system eq. (\ref{sht_1}). It will become clear
that, in this more general case, this condition it is also necessary for the
hyperbolicity of any reduced system, (\ref{red_sys}).

\begin{definition}
System (\ref{sht_1}) is called hyperbolic if there exists co-vector field $%
n_{a}$ such that

1- $\mathfrak{N}_{~\eta }^{Ab}n_{b}$ has not right kernel.

2- For each plane $l\left( \lambda \right) _{a}\in S_{n_{a}}^{\mathbb{C}}$,
if $\mathfrak{N}_{~\eta }^{Ab}l(\lambda )_{b}$ has non-trivial right kernel,
then $\lambda \in \mathbb{R}$.
\end{definition}

Condition 1, is necessary in order that there exists $h_{A}^{\gamma }$ such
that $h_{~A}^{\alpha }\mathfrak{N}_{~\alpha }^{Aa}n_{a}$ is invertible. It
also implies that the dimension of the left kernel of $\mathfrak{N}_{~\eta
}^{Ab}n_{b}$ is $c=e-u$. 

Condition 2 is obviously necessary as otherwise the exponent in the plane
wave solution would be real and for some values of $k_{a}$ would imply an
unbounded growth.

As shown in \cite{abalos2017necessary} the system is hyperbolic if and only
if for any fixed positive define pair of Hermitian forms 
$G_{AB}$  and $G^{\alpha \gamma }$ in $\Psi_L$ and $\Phi_R$ respectively,  the following polynomial equation in $%
\lambda $ and $\bar{\lambda}$ 
\begin{equation}
p\left( l\left( \lambda \right) _{a}\right) :=\det \left( G^{\alpha \gamma }%
\overline{\mathfrak{N}_{~\gamma }^{Aa}l\left( \lambda \right) _{a}}G_{AB}%
\mathfrak{N}_{~\eta }^{Bb}l\left( \lambda \right) _{b}\right) =0  \label{p_1}
\end{equation}%
has only real roots. This result does not depend on the particular pair $%
G_{AB}$ and $G^{\alpha \gamma }$ of metrics used, nevertheless we shall
need, to define uniformity, to choose any given, constant, pair of these
metrics.

In addition, for a given direction $n_{a}$ we shall call generalized
eigenvalues to the set of roots, $\lambda _{i}=\lambda _{i}\left( k\right) $%
, of eq. (\ref{p_1}) i.e. $p(l(\lambda _{i})_{a})=0$, and characteristic
co-vectors to the set of co-vectors $l\left( \lambda _{i}\right) _{a}\in
S_{n_{a}}$. So hyperbolicity means that eq. (\ref{symb_1}) has only real
generalized eigenvalues. They are the physical characteristics of the
system. That is, along them the physical degrees of freedom propagate.
Notice furthermore that the matrix $A_{~\gamma }^{\alpha a}k_{a}$ inherits
the generalized eigenvalues of the principal symbol. This matrix will in
general have further eigenvalues, as we shall show in the following
sections, they would depend on the particular reduction employed.

Our main theorem establishes which conditions on the principal symbol are
necessary and sufficient for system (\ref{sht_1}) to be strongly hyperbolic.
In order to formulate it, we first need to introduce further notation
related to the angles between vectorial subspaces. Introductions of these
topics are given, for instance, in \cite{taslaman2014principal} and \cite%
{afriat1957orthogonal}.

We call $\tau _{i}\left( k\right) $ with $i\in F_{\left( k\right) }=\left\{
1,...,w_{\left( k\right) }\right\} $ to the different eigenvalues\footnote{%
The $\tau _{i}\left( k\right) $ include to the generalized eigen-values $%
\lambda _{i}\left( k\right) $ and the other eigen-values obtained in the
reduction. Since all quantities depends on $k_{a}$ we explicitly put that dependence.}
of $A_{~\gamma }^{\alpha a}k_{a}$, and $\Phi _{R}^{\tau _{i}\left( k\right) }
$ , $\left( \Upsilon _{L}^{\tau _{i}\left( k\right) }\right) ^{\prime }$ to
the right vectorial eigen-subspace (see eq. (\ref{symb_1})) and the left
co-vectorial eigen-subspace\footnote{%
From now on, we denote the dual space associated to some space with the $%
^{\prime }$ symbol.} of $A_{~\gamma
}^{\alpha a}k_{a}$ respectively. 
Finally we call $\Upsilon _{L}^{\tau _{i}\left( k\right) }$ to the
subspace obtained by raising the index to the co-vectors of $\left( \Upsilon
_{L}^{\tau _{i}\left( k\right) }\right) ^{\prime }$ with $G^{\alpha \gamma }$. 
Now, since $\Upsilon _{L}^{\tau _{i}\left( k\right) }$ and $\Phi
_{R}^{\tau _{i}\left( k\right) }$ are subspaces of $\Phi _{R}$ and we
have the positive definite metric $G_{\alpha \gamma }$, it is possible to
define geometric angles between these subspaces that measure how close are
them to each other. The number of angles is equal to the smallest dimension
of the subspaces, since here $r_{\tau _{i}\left( k\right) }:=\dim \Upsilon
_{L}^{\tau _{i}\left( k\right) }=\dim \Phi _{R}^{\tau _{i}\left( k\right) }$%
, there are $r_{\tau _{i}\left( k\right) }$ angles.

With all the background given we are now in position to give another equivalent condition to the Kreiss Matrix theorem,
which is expressed in term of the angles between the subspaces. The proof
of this result is given in appendix \ref{appendix_a}, and just quote it here. 
A result by Strang \cite{strang1967strong} is used in the proof.
This condition will allow us to proof our main theorem \ref{Theorem_FL}.

\begin{theorem}
\label{Theorem_1}System (\ref{red_sys}) is strongly hyperbolic if and only
if it is hyperbolic with respect to $n_{a}$ and, for all $i\in F_{\left(
k\right) }$, and all $k_{a}$ non proportional to some $n_{a}$ with $%
\left\vert k\right\vert =1$, there is a constant angle $\vartheta <\frac{\pi 
}{2}$ such that all angles between \ $\Upsilon _{L}^{\tau _{i}\left(
k\right) }$ and $\Phi _{R}^{\tau _{i}\left( k\right) }$ are smaller than it.
\end{theorem}

For each $\tau _{i}\left( k\right) $, the $\cos$ of these angles turn to
be the $r_{\tau _{i}\left( k\right) }$ singular values of the square matrix 
\begin{equation}
\left( T^{\tau _{i}\left( k\right) }\right) _{~j}^{i}=\upsilon ^{i\alpha
}G_{\alpha \gamma }\delta \phi _{j}^{\gamma }  \label{eq_Tij_1}
\end{equation}%
where $\{\upsilon ^{i\alpha } \in \Upsilon _{L}^{\tau _{i}\left( k\right) }\}$
and $\{\delta \phi _{j}^{\gamma }\in \Phi _{R}^{\tau _{i}\left( k\right) }\}$,
with $i,j\in F_{\left( k\right) }$, are orthonormal bases of the
corresponding subspaces. 
Thus, these angles can be easily computed in examples.

We turn now to the main result of this work, obtaining necessary and
sufficient conditions for the existence of a well posed reduction. To that
end we shall use our previous theorem as a guidance. We shall consider the
right and left kernel of the complete (previous to any reduction) principal
symbol, project them, and obtain a condition between the angles of these subspaces.

We fix an $n_{a}$ for which the system is hyperbolic and consider the line $%
l(\lambda )_{a}\in S_{n_{a}}$. As before, the number and the geometric
multiplicity (i.e. the dimension of the right kernel) of the generalized
eigenvalues depends on $k_{a}.$ For each $k_{a}$ we shall call $\lambda
_{i}\left( k\right) $ with $i\in D_{\left( k\right) }:=\left\{
1,2,...,q_{\left( k\right) }\right\} $ to the different generalized
eigenvalues and $r_{\lambda _{i}\left( k\right) }$ to the geometric
multiplicity of the corresponding $\lambda _{i}\left( k\right) $. Notice
that in the quasi-linear case $\lambda _{i}\left( k\right) $, $q_{\left(
k\right) }$ and $r_{\lambda _{i}\left( k\right) }$ could depend on the
space-time points $x$ and the fields $\phi ^{\alpha }\left( x\right) $,
since we are considering the constant coefficients case this dependence does
not appear.

At each of the generalized eigenvalues, $\lambda _{i}\left( k\right) $, we
have a left and right kernel of the principal part. We shall call them $\Psi
_{L}^{\lambda _{i}\left( k\right) }$and $\Phi _{R}^{\lambda _{i}\left(
k\right) }$. They have dimensions $r_{\lambda _{i}\left( k\right) }+e-u$ and 
$r_{\lambda _{i}\left( k\right) }$ respectively. Using $\mathfrak{N}%
^{Aa}{}_{\alpha }n_{a}$ we can map $\Psi _{L}^{\lambda _{i}\left( k\right) }$
into $\Phi _{R}^{\prime }$ (the dual space of $\Phi _{R}$) and call it $%
\left( \Phi _{L}^{\lambda _{i}\left( k\right) }\right) ^{\prime }$. It is
not possible to know, in a generic way, its dimension. We only know that the
left kernel of $\mathfrak{N}^{Aa}{}_{\alpha }n_{a}$ is $e-u$, thus we can
bound the dimension as, $\ r_{\lambda _{i}\left( k\right) }\leq \dim \left(
\Phi _{L}^{\lambda _{i}\left( k\right) }\right) ^{\prime }\leq
e-u+r_{\lambda _{i}\left( k\right) }$.

These concepts allow to us, to state the result of \cite%
{abalos2017necessary}  in this simple equivalent form

\begin{equation}
\left. \left( \Phi _{L}^{\lambda _{i}\left( k\right) }\right) ^{\prime
}\right\vert _{\Phi _{R}^{\lambda _{i}\left( k\right) }}=\left( \Phi
_{R}^{\lambda _{i}\left( k\right) }\right) ^{\prime }.  \label{pp_1}
\end{equation}

Consider now the subspace obtained by rising the index to the elements of $%
\left( \Phi _{L}^{\lambda _{i}\left( k\right) }\right) ^{\prime }$ with $%
G^{\alpha \gamma }$, and calling that subspace $\Phi _{L}^{\lambda
_{i}\left( k\right) }$ we have, $\Phi _{L}^{\lambda _{i}\left( k\right)
}\subset \Phi _{R}$, and so we can define the angles between $\Phi
_{L}^{\lambda _{i}\left( k\right) }$ and $\Phi _{R}^{\lambda _{i}\left(
k\right) }$. For each $\lambda _{i}\left( k\right) $ and $k_{a}$ we call $%
\theta _{j}^{\lambda _{i}\left( k\right) }$ with $j\in I_{\lambda _{i}\left(
k\right) }:=\left\{ 1,...,r_{\lambda _{i}\left( k\right) }\right\} $ to
these angles, they are geometric quantities and the answer to our problem is
given in term of them. We are now in position to formulate our theorem:
\begin{theorem}
\label{Theorem_FL}The constant coefficient system (\ref{sht_1}) is strongly hyperbolic
(admits at least one hyperbolizer) if and only if it is hyperbolic with
respect to some direction $n_{a}$ and, for all $i\in D_{\left( k\right) }$
and all $k_{a}$ non proportional to $n_{a}$, with $\left\vert k\right\vert =1$,
there is a constant maximum angle $\vartheta <\frac{\pi }{2}$ between \ $%
\Phi _{L}^{\lambda _{i}\left( k\right) }$ and $\Phi _{R}^{\lambda _{i}\left(
k\right)}$.
\end{theorem}
It is possible to show that the necessary condition (\ref{pp_1}) implies
that for each $k_a$ there exists a metric such that all the angles between $\Phi _{L}^{\lambda _{i}\left( k\right) }$ and $\Phi
_{R}^{\lambda _{i}\left( k\right) }$ vanish. But this is
not sufficient for the theorem, since we need a global metric (independent on $k$) such that a lower order bound (the existence of $\vartheta $) exists. Written in term of the cosine of
the angles 
\begin{equation}
\min_{i\in D_{\left( k\right) }\text{ }j\in I_{\lambda _{i}\left( k\right) }%
\text{ \ }\left\vert k\right\vert =1\text{ }}\cos \theta _{j}^{\lambda
_{i}\left( k\right) }\geq \cos \vartheta >0  \label{angles_1}
\end{equation}%
for all $k_{a}$ non proportional to $n_{a}$.

The angles are computed in the same way as before. Using orthonormal bases
of the spaces $\Phi _{L}^{\lambda _{i}\left( k\right) }$ and $\Phi
_{R}^{\lambda _{i}\left( k\right) }$ and building a new matrix $\left(
T^{\tau _{i}\left( k\right) }\right) _{~j}^{I}$ resulting of contracting
that bases with the metric. 
Notice that this matrix will in general be rectangular,
since any base of $\Phi _{L}^{\lambda _{i}\left( k\right) }$ has $r_{\lambda
_{i}\left( k\right) }$ or more vectors.

Assuming that the conditions of theorem \ \ref{Theorem_FL} holds, we shall
show how to build reductions $h_{~A}^{\alpha}$ in such a way that the
degeneracy of the generalized eigenvalues does not change, and the new
eigenvalues introduced by $h_{~A}^{\alpha}$ can be chosen to be simple
and different to the ones of the whole system. 
These reduction will comply with the the
hypothesis of theorem \ref{Theorem_1} and so conclude the proof.


\section{\label{Proof_theorems} Theorem Proof}

The proof of the theorem is split into five subsection. First we introduce,
in subsection \ref{subsec_Kronecker_decomp}, the Kronecker decomposition of
pencils, this decomposition will be applied to the principal symbol and
later used to build certain reductions. Next, in subsection \ref%
{Necessary_condition_1}, after introducing notation about basis of the
corresponding subspaces of left and right kernels of the principal symbol, we
proof Lemma \ref{Lemma_necessary_condition} which connects the hypothesis
in our theorem with the necessary condition in \cite{abalos2017necessary}.
In subsection \ref{Proof_1}, we proof Lemma \ref{main_lema} which gives a
set of equivalent necessary and sufficient conditions for the existence of
reductions $h_{~A}^{\alpha }$ which gives diagonalizable$A_{~\gamma
}^{\alpha a}k_{a}$, and explicitly display all possible reductions. In
subsection \ref{theo_well_posed} we complete the proof by showing, in Lemma %
\ref{lemma_2}, that it is possible to find reductions such that the extra
eigenvalues of $A_{~\gamma }^{\alpha a}k_{a}$, the propagation velocities of
the constraints, are simple and different to the physical propagations
velocities. Finally in subsection \ref{theo_well_posed_II}, we further
restrict the reductions, $h_{~A}^{\alpha }$, so that the angles between $%
\Phi _{L}^{\lambda _{i}\left( k\right) }$, $\Phi _{R}^{\lambda _{i}\left(
k\right) }$ and $\Upsilon _{L}^{\lambda _{i}\left( k\right) }$, $\Phi
_{R}^{\lambda _{i}\left( k\right) }$ are equal. Applying theorem \ref%
{Theorem_1} to the reduced system we conclude the proof.

\subsection{Kronecker decomposition of pencils\label{subsec_Kronecker_decomp}}


Consider now the principal symbol matrix pencil equation 
\begin{equation}
\mathfrak{N}_{~\eta }^{Ab}l\left( \lambda \right) _{b}=\lambda \left( -%
\mathfrak{N}_{~\eta }^{Ab}n_{b}\right) +\left( \mathfrak{N}_{~\eta
}^{Ab}k_{b}\right)   \label{symbol_2}
\end{equation}

for fixed $n_{a}$ and $k_{a}$ in the line $l\left( \lambda \right) _{b}\in
S_{n_{a}}$. The intrinsic structure of this matrix pencil, at each one of
these points, determines the strong hyperbolicity of the system. This
structure will become apparent via the Kronecker decomposition \cite%
{gantmacher1992theory}, \cite{gantmakher1998theory}. It consists on change
of bases for the field and equation spaces, which depend on $k_{a}$, $n_{a}$%
, but are independent of the parameter $\lambda $.~\footnote{%
Actually, since the symbol is linear in $k_{a}$ the basis can only
depend on the direction of $k_{a}$ and not on its magnitude, that is
the reason why we only consider lines where $k_{a}$ is of unit length
for some arbitrary metric.} The new bases transform the symbol (\ref%
{symbol_2}) into simple blocks. In this form, the study of right and left
kernel became easy.

Since the hyperbolicity condition restrict $\mathfrak{N}_{~\eta }^{Ab}n_{b}$
to have no right kernel, the allowed blocks of the Kronecker decomposition
of (\ref{symbol_2}) simplifies into just two types of possible blocks,
namely:

a) $J_{m}\left( \lambda_{i}\right) -$Jordan blocks:

\begin{equation*}
J_{m}\left( \lambda_{i}\right) =\left( 
\begin{array}{cccc}
\lambda-\lambda_{i} & 1 & 0 & 0 \\ 
0 & \lambda-\lambda_{i} & ... & 0 \\ 
0 & 0 & ... & 1 \\ 
0 & 0 & 0 & \lambda-\lambda_{i}%
\end{array}
\right) \in \mathbb{C} ^{m\times m}
\end{equation*}
with $\lambda_{i}$ the generalized eigenvalues introduced in section \ref%
{Strongly};

b) $L_{m}^{T}-$ blocks:

\begin{equation*}
L_{m}^{T}=\left( \overset{m}{\overbrace{%
\begin{array}{cccc}
\lambda & 0 & 0 & 0 \\ 
1 & \lambda & 0 & 0 \\ 
0 & 1 & ... & 0 \\ 
0 & 0 & ... & \lambda \\ 
0 & 0 & 0 & 1%
\end{array}
}}\right) \in%
\mathbb{C}
^{m+1\times m}
\end{equation*}

, the $L_{0}^{T}-$ case are vanishing rows.

As an example 
\begin{equation}
\mathfrak{N}_{~\eta }^{Ab}l\left( \lambda \right) _{b}=\left( 
\begin{array}{cccc}
J_{1}\left( \lambda _{1}\right)  & 0 & 0 & 0 \\ 
0 & J_{3}\left( \lambda _{2}\right)  & 0 & 0 \\ 
0 & 0 & L_{2} & 0 \\ 
0 & 0 & 0 & L_{1} \\ 
0 & 0 & 0 & 0 \\ 
0 & 0 & 0 & 0%
\end{array}%
\right) \in \mathbb{C}^{11\times 7}  \label{example_1}
\end{equation}%
%
In general the Kronecker decomposition includes other blocks (see appendix %
\ref{appendix_II}), which do not appear here.

The Kronecker structure of a particular symbol is unique, however, in
general, there will exist different bases that lead to it. In appendix \ref%
{appendix}\ we shall show how to find the different blocks of $\mathfrak{N}%
_{~\eta }^{Ab}l\left( \lambda \right) _{b}$.

It is important to notice that most physical systems have only $L_{1}^{T}$-blocks and $L_{0}^{T}$-rows. That is the case for instance for Maxwell Electrodynamics, Non
linear electrodynamics, Force Free electrodynamics, ideal, charged and
conformal fluids, etc. Indeed, all constraints in Geroch's sense \cite{geroch1996partial} are related just to $L_{0}^{T}$-rows and $L_{1}^{T}$-blocks. 

\subsection{Necessary condition for strong hyperbolicity\label{Necessary_condition_1}}

Let be $l\left( \lambda \right) _{b}=-\lambda n_{a}+k_{a}\in S_{n_{a}}$, the
Kronecker decomposition asserts that the subspace $\Psi _{L}$ of left kernel
of the principal symbol $\mathfrak{N}_{~\eta }^{Bb}l\left( \lambda \right)
_{b}$ is expanded by a set of $e-u$ unique vectors $\chi _{A}^{s}\left(
\lambda \right) $ where $s\in C_{\chi }=\left\{ 1,..,e-u\right\} $ for any $%
\lambda $, and it increases when $\lambda =\lambda _{i}\left( k\right) $. We
shall choose a set of arbitrary $r_{\lambda _{i}\left( k\right) }$ new
co-vectors $\left( \upsilon _{\lambda _{i}}^{l}\right) _{A}$ with $l\in
I_{\lambda _{i}\left( k\right) }:=\left\{ 1,...,r_{\lambda _{i}\left(
k\right) }\right\} $ to complete a base for $\Psi _{L}^{\lambda _{i}}$. We
shall refer to $\left( \upsilon _{\lambda _{i}}^{l}\right) _{A}$ as the
generalized eigen-covectors. Thus, for each $\lambda =\lambda _{i}\left(
k\right)$, 
\begin{equation}
\Psi _{L}^{\lambda _{i}\left( k\right) }=span\{(\chi _{\lambda _{i}}^{s})_{A}%
\text{, }(\upsilon _{\lambda _{i}}^{l})_{A}\}\text{ with }s\in C_{\chi }%
\text{ and }l\in I_{\lambda _{i}\left( k\right) }  \label{left_kernel_1}
\end{equation}
where $\left( \chi _{\lambda _{i}}^{s}\right) _{A}:=\chi _{A}^{s}\left(
\lambda _{i}\left( k\right) \right) $.

On the other hand, $\mathfrak{N}_{~\eta }^{Bb}l\left( \lambda \right) _{b}$
only have right kernel when $\lambda =\lambda _{i}\left( k\right) $, this
subspace $\Phi _{R}^{\lambda _{i}\left( k\right) }$ is expanded by 
\begin{equation}
\Phi _{R}^{\lambda _{i}\left( k\right) }=span\{(\delta \phi _{j}^{\lambda
_{i}})^{\alpha }\}\text{ with \ }j\in I_{\lambda _{i}\left( k\right) }
\end{equation}%
Notice that $\dim \Psi _{L}^{\lambda _{i}\left( k\right) }=e-u+r_{\lambda
_{i}\left( k\right) }$ and $\dim \Phi _{R}^{\lambda _{i}\left( k\right)
}=r_{\lambda _{i}\left( k\right) }$.

As in \cite{abalos2017necessary} we shall now look at the Singular Value
Decomposition (SVD) of the principal symbol at $\lambda =\lambda _{i}$. For
this we introduce two scalar products shall use the positive definite Hermitian forms $G_{AB}$ and 
$G_{\alpha \gamma }$ on each of the spaces. From the discussion above there
will be $r_{\lambda _{i}\left( k\right) }$ vanishing singular values,%
\begin{equation*}
\sigma _{u+1-j}\left[ \mathfrak{N}_{~\alpha }^{Ab}l\left( \lambda _{i}\left(
k\right) \right) _{a}\right] =0
\end{equation*}%
with~\footnote{%
Recall that the singular values are ordered in such a way that $\sigma _{1}%
\left[ \mathfrak{N}_{~\alpha }^{Ab}n\left( \lambda _{i}\left( k \right)
\right) _{a}\right] \geq \sigma _{2}\left[ \mathfrak{N}_{~\alpha
}^{Ab}n\left( \lambda _{i}\left( k \right) \right) _{a}\right] \geq
...\geq \sigma _{u}\left[ \mathfrak{N}_{~\alpha }^{Ab}n\left( \lambda
_{i}\left( k \right) \right) _{a}\right] $.} $j\in $ $I_{\lambda
_{i}\left( k\right) }$.

Consider now the extended two-parameter line $l_{\varepsilon ,\theta }\left(
\lambda \right) _{a}=-\varepsilon e^{i\theta }n_{a}+l\left( \lambda \right)
_{a}$ with $\varepsilon $ real, $\theta \in \left[ 0,2\pi \right] $, $0\leq
\left\vert \varepsilon \right\vert <<1$. \ and with $l\left( \lambda \right)
_{a}\in S_{n_{a}}.$ As it is shown in \cite{moro2002first}\ and \cite%
{soderstrom1999perturbation} 
\begin{equation}
\sigma _{u+1-j}\left[ \mathfrak{N}_{~\alpha }^{Ab}l_{\varepsilon ,\theta
}\left( \lambda _{i}\left( k\right) \right) _{a}\right] =\left( \rho
_{\lambda _{i}\left( k\right) }\right) _{j}\varepsilon +O\left( \varepsilon
^{2}\right)   \label{chi_eq_2}
\end{equation}%
where \ $\left( \rho _{\lambda _{i}\left( k\right) }\right) _{j}$ are the
singular values of the matrix 
\begin{equation*}
\left( R^{\lambda _{i}\left( k\right) }\right) _{~m}^{I}=\left( 
\begin{array}{c}
\left( \hat{\upsilon}_{\lambda _{i}}^{l}\right) _{A} \\ 
\left( \hat{\chi}_{\lambda _{i}}^{s}\right) _{A}%
\end{array}%
\right) \mathfrak{N}_{~\alpha }^{Ab}n_{b}\left( \widetilde{\delta \phi }%
_{m}^{\lambda _{i}}\right) ^{\alpha }
\end{equation*}%
with $I=\left( l,s\right) $, and where $\left( \hat{\upsilon}_{\lambda
_{i}}^{l}\right) _{A}$, $\left( \hat{\chi}_{\lambda _{i}}^{s}\right) _{A}$
is an orthonormalized base with respect to the metric $G^{AB}$, of $\Psi
_{L}^{\lambda _{i}\left( k\right) }$ and $\left( \widetilde{\delta \phi }%
_{m}^{\lambda _{i}}\right) ^{\alpha }$ is a orthonormalized base, with
respect to the metric $G_{\alpha \gamma }$, of $\Phi _{R}^{\lambda
_{i}\left( k\right) }$.

In \cite{abalos2017necessary} it was shown that a necessary condition for
the existence of a reduction is that the singular values of $\mathfrak{N}%
_{~\alpha }^{Ab}l_{\varepsilon ,\theta }\left( \lambda _{i}\left( k\right)
\right) _{a}$ are either of order $O\left( \varepsilon ^{0}\right) $ or $%
O\left( \varepsilon ^{1}\right) $ which is equivalent to the condition that
none of singular values $\left( \rho _{\lambda _{i}\left( k\right) }\right)
_{j}$ of $R^{\lambda _{i}\left( k\right) }$ vanish. In the following lemma
we shall proof that condition (\ref{angles_1}) in theorem \ref{Theorem_FL} \
implies that $\left( \rho _{\lambda _{i}\left( k\right) }\right) _{j}>0$ and
so that the necessary condition in \cite{abalos2017necessary} holds.

\begin{lemma}
\label{Lemma_necessary_condition}For all $i\in D_{\left( k\right) }$, $j\in
I_{\lambda _{i}\left( k\right) }$, with $k_a$ not
proportional to $n_{a}$ and $\left\vert k\right\vert =1$ , if \ equation \ref{angles_1} holds, then $\left(
\rho _{\lambda _{i}\left( k\right) }\right) _{j}>0$.
\end{lemma}

\begin{proof}
Recalling that $\Phi _{L}^{\lambda _{i}\left( k\right) }$ is the map of $%
\Psi _{L}^{\lambda _{i}\left( k\right) }$ into $\Phi _{R}$ by $\mathfrak{N}%
_{~\alpha }^{Aa}n_{a}G^{\alpha \gamma }$, so a particular set that
spans this subspace is $\left( \hat{\chi}_{\lambda _{i}}^{s}\right) ^{\gamma
}:=\left( \hat{\chi}_{\lambda _{i}}^{s}\right) _{A}\mathfrak{N}_{~\alpha
}^{Aa}n_{a}G^{\alpha \gamma }$ and $\left( \hat{\upsilon}_{\lambda
_{i}}^{l}\right) ^{\gamma }:=\left( \hat{\upsilon}_{\lambda _{i}}^{l}\right)
_{A}\mathfrak{N}_{~\alpha }^{Aa}n_{a}G^{\alpha \gamma }$ i.e.
\begin{equation}
\Phi _{L}^{\lambda _{i}\left( k\right) }=span\{\left( \hat{\chi}_{\lambda
_{i}}^{s}\right) ^{\gamma }\text{, }\left( \hat{\upsilon}_{\lambda
_{i}}^{l}\right) ^{\gamma }\}\text{ with }s\in C_{\chi }\text{ and }l\in
I_{\lambda _{i}\left( k\right) }.
\end{equation}
Notice that there might be linearly dependent vectors among the $\left( \hat{%
\chi}_{\lambda _{i}}^{s}\right) ^{\gamma }$ and $\left( \hat{\upsilon}_{\lambda
_{i}}^{l}\right) ^{\gamma }$,   
 however, consider first the case that they are linear independent.
To calculate the angles $\theta
_{k}^{\lambda _{i}\left( k\right) }$ between $\Phi _{L}^{\lambda _{i}\left(
k\right) }$ and $\Phi _{R}^{\lambda _{i}\left( k\right) }$ we need to use
orthonormalized basis on these subspaces. Calling $Q^{J}{}_{I}$, with $%
J=\left( s,l\right) $ and $I=\left( m,n\right) $, to the square matrix that
connects the basis $\left( \hat{\chi}_{\lambda _{i}}^{s}\right) ^{\gamma }$, 
$\left( \hat{\upsilon}_{\lambda _{i}}^{l}\right) ^{\gamma }$ with a new
orthonormalized basis $\left( \tilde{\chi}_{\lambda _{i}}^{s}\right)
^{\gamma }$, $\left( \tilde{\upsilon}_{\lambda _{i}}^{l}\right) ^{\gamma }$,
in the metric $G_{\alpha \gamma }$, of $\Phi _{L}^{\lambda _{i}\left(
k\right) }$,

\begin{equation*}
\left( 
\begin{array}{c}
\left( \tilde{\chi}_{\lambda _{i}}^{s}\right) ^{\gamma } \\ 
\left( \tilde{\upsilon}_{\lambda _{i}}^{l}\right) ^{\gamma }%
\end{array}%
\right) =Q_{~I}^{J}\left( 
\begin{array}{c}
\left( \hat{\chi}_{\lambda _{i}}^{m}\right) ^{\gamma } \\ 
\left( \hat{\upsilon}_{\lambda _{i}}^{n}\right) ^{\gamma }%
\end{array}%
\right) 
\end{equation*}%
The cosines of the angles $\cos \theta _{k}^{\lambda _{i}\left( k\right) }$
between $\Phi _{L}^{\lambda _{i}\left( k\right) }$ and $\Phi _{R}^{\lambda
_{i}\left( k\right) }$ are the singular values of the matrix%
\begin{equation} 
\left( T^{\lambda _{i}\left( k\right) }\right) _{~j}^{J}=\left( 
\begin{array}{c}
\left( \tilde{\chi}_{\lambda _{i}}^{s}\right) ^{\gamma } \\ 
\left( \tilde{\upsilon}_{\lambda _{i}}^{l}\right) ^{\gamma }%
\end{array}%
\right) G_{\gamma \eta }\left( \widetilde{\delta \phi }_{m}^{\lambda
_{i}}\right) ^{\eta }=Q_{~I}^{J}\left( R^{\lambda _{i}\left( k\right)
}\right) _{~j}^{I} \label{abc}  
\end{equation}%
Since the singular values of $T^{\tau _{i}\left( k\right) }$ do not vanish
by hypothesis (eq. (\ref{angles_1})), then 
 $R^{\lambda _{i}\left( k\right) }$ has not right kernel and their singular values  can not vanish. 

In the case of the $\left( \hat{%
\chi}_{\lambda _{i}}^{s}\right) ^{\gamma }$ and $\left( \hat{\upsilon}_{\lambda
_{i}}^{l}\right) ^{\gamma }$ are linear dependent, some of them shall be
removed until obtaining a base for define $\left( T^{\lambda _{i}\left( k\right) }\right) _{~j}^{J}$. Next,  there exists a rectangular matrix $Q_{~I}^{J}$ such that equation
(\ref{abc}) holds and the conclusion is the same. 
Thus,
we conclude the proof of the lemma.
\end{proof}


\subsection{\label{hyperbolizer}Building reductions\label{Proof_1}}

In this subsection we proof a lemma which gives a set of equivalent
conditions and furthermore shows how to build, using the Kronecker
decomposition of the principal part, the general reduction $h_{A}^{\alpha }$
giving a diagonalizable reduced matrix. It is important to notice that if any $J_{m}$
Jordan block, with $m\geq 2$ appears in the Kronecker decomposition, then
the system is intrinsically weakly hyperbolic. However, if that blocks do not appear, this condition it is
not sufficient for strong hyperbolicity since two problems can be present.
The first one is that a reduction can introduced a $J_{m}$ Jordan block with $m$ 
$\geq 2$, from a $L^{T}$ block. This $L^{T}$ block will be associated to constraints
propagation, reducing the system to a weakly hyperbolic one. This would give
an ill posed subsidiary system for the constraint propagation. The second
one, is that of a reduction for which $A_{~\gamma }^{\alpha a}k_{a}$ is
diagonalizable, but not uniformly diagonalizable, then the systems will also
be ill posed. To solve these problems we shall use, in the next subsections,
the results of Lemma \ref{lemma_2} and the lower bound condition eq. (\ref%
{angles_1}).

\begin{lemma}
\label{main_lema}Let system (\ref{sht_1}) be hyperbolic for $n_{a}$, then
the following conditions are equivalent: \ For each line $l\left( \lambda
\right) _{b}=-\lambda n_{a}+k_{a}$ in $S_{n_{a}}$

$i)$ There exists a reduction $h_{A}^{\alpha }$, homogeneous of degree $0$
in $k_{a}$, such that $A_{~\gamma }^{\alpha a}k_{a}$ is diagonalizable.

$ii)$ The Kronecker Decomposition of the principal symbol pencil, (\ref%
{symbol_2}) has all their Jordan blocks of dimension $1 $.

$iii)$ The Singular Value Decomposition of the principal symbol pencil \newline $%
\mathfrak{N}_{~\alpha }^{Aa}l_{\varepsilon ,\theta }\left( \lambda \right)
_{a}$ along any extended line $l_{\varepsilon ,\theta }\left( \lambda
\right) _{a}$ has only singular values of orders $O\left( \left\vert
\varepsilon \right\vert ^{0}\right) $ and $O\left( \left\vert \varepsilon
\right\vert ^{1}\right) $ i.e. $\left( \rho _{\lambda _{i}\left( k\right)
}\right) _{j}>0$ for all $i\in D_{\left( k\right) },\text{ and }j\in
I_{\lambda _{i}\left( k\right) }$.
\end{lemma}

\begin{proof}
The Kronecker decomposition  of the principal symbol \cite{gantmacher1992theory}, \cite%
{gantmakher1998theory} is 
\begin{align}
\mathfrak{N}_{~\eta }^{Ab}l\left( \lambda \right) _{b}& =\lambda \left( -%
\mathfrak{N}_{~\eta }^{Ab}n_{b}\right) +\left( \mathfrak{N}_{~\eta
}^{Ab}k_{b}\right)   \label{MR_Kro_1} \\
& =Y_{~B}^{A}\left( x,\phi ,n,k\right) K_{~\alpha }^{B}\left( \lambda
\right) W_{~\eta }^{\alpha }\left( x,\phi ,n,k\right)   \label{MR_Kro}
\end{align}%
where $l\left( \lambda \right) _{b}=-\lambda n_{b}+k_{b}\in S_{n_{a}},$ $Y$, 
$W$ are invertible operators and%
\begin{equation*}
K_{~\alpha }^{B}\left( \lambda \right) :=\lambda I_{~\alpha }^{B}+M_{~\alpha
}^{B}
\end{equation*}%
is the Kronecker matrix. It is a block matrix with $J_{m}\left( \lambda
_{i}\right) -$Jordan blocks, $L_{i}^{T}$ blocks and trivial $L_{0}^{T}-$
rows. The operators $I$ and $M$ are unique but in general they could change
for different values of $\left( n,k\right) $ \footnote{$\left( x,\phi
,n ,k \right) $ in the quasilinear case.}, although that does seems
not occur for the standard physical examples. However, in general $Y$ and $W$
could be chosen in different ways even at the same point. Notice that, 
\begin{equation*}
-\mathfrak{N}_{~\eta }^{Ab}n_{b}=Y_{~B}^{A}I_{~\alpha }^{B}W_{~\eta
}^{\alpha }\text{ and }\mathfrak{N}_{~\eta }^{Ab}k_{b}=Y_{~B}^{A}M_{~\alpha
}^{B}W_{~\eta }^{\alpha }.
\end{equation*}

$ii\Rightarrow i$

We propose the following ansatz for a reduction 
\begin{equation}
h_{~C}^{\rho}=S^{\rho\alpha}\overline{W_{~\alpha}^{\delta}}H_{\delta
C}\left( Y^{-1}\right) _{~D}^{C}   \label{h_1}
\end{equation}
with $S$ being any invertible bilinear form and $H$ another one that depends
on the explicit form of $K$, and which will be given explicitly later on.

With this ansatz the reduced system simplifies to:%
\begin{equation*}
h_{~C}^{\rho }\mathfrak{N}_{~\eta }^{Cb}l\left( \lambda \right) _{b}=\lambda
\left( S^{\rho \alpha }\overline{W_{~\alpha }^{\delta }}H_{\delta
C}I_{~\gamma }^{C}W_{~\eta }^{\gamma }\right) +\left( S^{\rho \alpha }%
\overline{W_{~\alpha }^{\delta }}H_{\delta C}M_{~\gamma }^{C}W_{~\eta
}^{\gamma }\right) 
\end{equation*}

Thus, assuming for a moment that:

\textbf{a)} $H_{\delta C}I_{~\alpha}^{C}$ is a positive definite Hermitian
form and,

\textbf{b)} $H_{\delta C}M_{~\alpha}^{C}$ is a Hermitian form,

We conclude that $\overline{W_{~\alpha }^{\delta }}H_{\delta C}I_{~\gamma
}^{C}W_{~\eta }^{\gamma }$, and $\overline{W_{~\alpha }^{\delta }}H_{\delta
C}M_{~\gamma }^{C}W_{~\eta }^{\gamma }$ are Hermitian forms being the first
positive definite. From this we have,%
\begin{align}
A_{~\eta }^{\upsilon a}k_{a}& =-\left( \left( S\overline{W_{~}}HIW\right)
^{-1}\right) _{~\tau }^{\upsilon }\left( S^{\tau \alpha }\overline{%
W_{~\alpha }^{\delta }}H_{\delta C}M_{~\gamma }^{C}W_{~\eta }^{\gamma
}\right)   \notag \\
& =-\left( W^{-1}\right) _{~\gamma }^{\upsilon }\left( \left( H_{\delta
C}I_{~\gamma }^{C}\right) ^{-1}\right) ^{\gamma \rho }H_{\delta C}M_{~\gamma
}^{C}W_{~\eta }^{\gamma }.  \label{A_1}
\end{align}%
Where $\left( \left( S\overline{W_{~}}HIW\right) ^{-1}\right) _{~\tau
}^{\upsilon }=\left( \left( S^{\tau \alpha }\overline{W_{~\alpha }^{\delta }}%
H_{\delta C}I_{~\gamma }^{C}W_{~\upsilon }^{\gamma }\right) ^{-1}\right)
_{~\tau }^{\upsilon }$. Thus, $A_{~\eta }^{\upsilon a}k_{a}$ is Hermitizable
(or symmetrizable), and therefore diagonalizable. Furthermore it has only
real eigenvalues. Notice that $S$ introduces more degrees of freedom in $%
h_{~C}^{\rho }$ that can be chosen arbitrarily as long as $S$ is invertible.
But they turn not to be relevant, since they do not appears in  $A_{~\eta }^{\upsilon a}k_{a}$.

Thus, if we find $H$ satisfying a) and b) the implication $ii\Rightarrow i$
will be proven.

Using the hypothesis $ii$ in theorem we shall conclude $i$ by building $%
H_{\delta C}$. We shall propose a specific $H$\footnote{The most general form of $H$, with a reduced principal symbol not necessarily diagonalizable, is presented in the ADM example.} for each blocks of $K_{~\eta
}^{A}\left( \lambda \right) $. Consider first the Jordan blocks. The$\
J_{m}\left( \lambda _{i}\right) -$Jordan blocks have kernel when $\lambda $
is equal to the generalized eigenvalues $\lambda _{i}$, then hyperbolicity
with respect to $n_{a}$ implies that these eigenvalues are real. In
addition, condition $ii$ implies these blocks are 1-dimensional, therefore
in $K_{~\eta }^{A}\left( \lambda \right) $ there appears $m\times m$
identity blocks multiplied by $\left( \lambda -\lambda _{i}\right) $ that we
call $Id_{m}\left( \lambda -\lambda _{i}\right) .$ For these blocks we
choose $H_{Id_{m}}$ to be any positive definite Hermitian form of size $%
m\times m$.

We now turn to the generic $\left( L_{m}^{T}\right) _{~j}^{i}$ blocks with $%
i=1,...,m+1$ and $j=1,....,m$. We propose a particular $\left(
H_{L_{m}^{T}}\right) _{si}$ with $s=1,...,m$ of the following form: 
\begin{equation}
\left( H_{L_{m}^{T}}\right) _{si}\left( L_{m}^{T}\right) _{~j}^{i}\left(
\lambda \right) =\left( 
\begin{array}{ccccc}
a_{1} & a_{2} & ... & a_{m} & a_{m+1} \\ 
a_{2} & ... & a_{m} & a_{m+1} & a_{m+2} \\ 
... & a_{m} & a_{m+1} & a_{m+2} & ... \\ 
a_{m} & a_{m+1} & a_{m+2} & ... & a_{2m}%
\end{array}%
\right) \left( 
\begin{array}{cccc}
\lambda  & 0 & 0 & 0 \\ 
1 & \lambda  & 0 & 0 \\ 
0 & 1 & ... & 0 \\ 
0 & 0 & ... & \lambda  \\ 
0 & 0 & 0 & 1%
\end{array}%
\right)   \label{HLm}
\end{equation}%
with all components real. \ Notice that $L_{m}^{T}$ can be split into $%
L_{m}^{T}=\left( \lambda I_{L_{m}^{T}}+M_{L_{m}^{T}}\right) $ with 
\begin{equation*}
I_{L_{m}^{T}}=\left( 
\begin{array}{cccc}
1 & 0 & 0 & 0 \\ 
0 & 1 & 0 & 0 \\ 
0 & 0 & ... & 0 \\ 
0 & 0 & 0 & 1 \\ 
0 & 0 & 0 & 0%
\end{array}%
\right) \text{ and }M_{L_{m}^{T}}=\left( 
\begin{array}{cccc}
0 & 0 & 0 & 0 \\ 
1 & 0 & 0 & 0 \\ 
0 & 1 & 0 & 0 \\ 
0 & 0 & ... & 0 \\ 
0 & 0 & 0 & 1%
\end{array}%
\right) .
\end{equation*}

Then%
\begin{align*}
\left( H_{L_{m}^{T}}\right) _{si}\left( L_{m}^{T}\right) _{~j}^{i}\left(
\lambda \right) = & \left( 
\begin{array}{ccccc}
a_{1} & a_{2} & ... & a_{m} & a_{m+1} \\ 
a_{2} & ... & a_{m} & a_{m+1} & a_{m+2} \\ 
... & a_{m} & a_{m+1} & a_{m+2} & ... \\ 
a_{m} & a_{m+1} & a_{m+2} & ... & a_{2m}%
\end{array}
\right) \left( \lambda I_{L_{m}^{T}}+M_{L_{m}^{T}}\right) \\
 =&\lambda\left( 
\begin{array}{cccc}
a_{1} & a_{2} & ... & a_{m} \\ 
a_{2} & ... & a_{m} & a_{m+1} \\ 
... & a_{m} & a_{m+1} & ... \\ 
a_{m} & a_{m+1} & ... & a_{2m-1}%
\end{array}
\right) +\left( 
\begin{array}{cccc}
a_{2} & ... & a_{m} & a_{m+1} \\ 
... & a_{m} & a_{m+1} & a_{m+2} \\ 
a_{m} & a_{m+1} & a_{m+2} & ... \\ 
a_{m+1} & a_{m+2} & ... & a_{2m}%
\end{array}
\right) \\
 = &\lambda\left( g_{m}\right) _{\delta\alpha}+\left( l_{m}\right)
_{\delta\alpha}
\end{align*}
with 
\begin{equation}
\left( g_{m}\right) _{\delta\alpha}=\left( 
\begin{array}{cccc}
a_{1} & a_{2} & ... & a_{m} \\ 
a_{2} & ... & a_{m} & a_{m+1} \\ 
... & a_{m} & a_{m+1} & ... \\ 
a_{m} & a_{m+1} & ... & a_{2m-1}%
\end{array}
\right) \text{ and }\left( l_{m}\right) _{\delta\alpha}=\left( 
\begin{array}{cccc}
a_{2} & ... & a_{m} & a_{m+1} \\ 
... & a_{m} & a_{m+1} & a_{m+2} \\ 
a_{m} & a_{m+1} & a_{m+2} & ... \\ 
a_{m+1} & a_{m+2} & ... & a_{2m}%
\end{array}
\right)   \label{g_l}
\end{equation}

Notice the cascade form of $H_{L_{m}^{T}}$ and that both, $g_{m}$, and $l_{m}
$ are symmetric. It is no so difficult to see that this is the most general
form of $H_{L_{m}^{T}}$ fulfilling that condition. To hold conditions a)
and b) we only need to show that $g_{m}$ can be chosen to be positive definite
by choosing appropriately the coefficients $a_{i}$ with $i=1,...,2m-1.$ We
do this by induction in $m$. When $m=1$ the positivity condition is just $%
a_{1}>0$. Assuming the inductive hypothesis: $g_{m}$ is positive definite
(as in eq. (\ref{g_l})), we enlarge the Hermitian form to $g_{m+1}$ by
adding a new column and a new row 
\begin{equation*}
\left( g_{m+1}\right) _{\delta\alpha}=\left( 
\begin{array}{ccccc}
a_{1} & a_{2} & ... & a_{m} & a_{m+1} \\ 
a_{2} & ... & a_{m} & a_{m+1} & ... \\ 
... & a_{m} & a_{m+1} & ... & a_{2m-1} \\ 
a_{m} & a_{m+1} & .... & a_{2m-1} & a_{2m} \\ 
a_{m+1} & .... & a_{2m-1} & a_{2m} & a_{2m+1}%
\end{array}
\right) . 
\end{equation*}

Thus there appear just two new coefficients $a_{2m}$ and $a_{2m+1}$ (they
are the only new coefficient in $g_{m+1}$ that are not in $g_{m}$). We need
to show that there exists a possible choice of these coefficients such that $%
g_{m+1}$ is positive definite. Since $g_{m}$ as in eq. (\ref{g_l}) is
positive definite, by Sylvester's criterion, we only need to show that $%
\det\left( g_{m+1}\right) >0.$ Expanding the determinant along the last
column the condition becomes, 
\begin{equation*}
\det\left( g_{m+1}\right) =a_{2m+1}\det\left( g_{m}\right) +f\left(
a_{1},...,a_{2m}\right) >0 
\end{equation*}
for some function $f$ that does not depend on $a_{2m+1}.$ Thus, choosing any 
$a_{2m}$, $f\left( a_{1},...,a_{2m}\right) $ becomes known, and since $%
\det\left( g_{m}\right) >0$ we just need to take $a_{2m+1}$ so that, 
\begin{equation*}
a_{2m+1}>-\frac{f\left( a_{1},...,a_{2m}\right) }{\det\left( g_{m}\right) }
\end{equation*}
to obtain a positive definite $g_{m+1}$.

Finally, for the vanishing rows of $K_{~\eta}^{A}\left( \lambda\right) $,
the $L_{0}^{T}$ rows, we could choose for $H$ arbitrary columns. They do not
seem to play any role.

The resulting structure for $H$ becomes as shown in the following example:%
\begin{align}
& \left( H\right) _{\delta A}K_{~\eta}^{A}\left( \lambda\right) = \notag \\
& \left( 
\begin{array}{cccccc}
H_{Id_{m_{3}}} & 0 & 0 & 0 & A_{1} & B_{1} \\ 
0 & H_{Id_{m_{4}}} & 0 & 0 & A_{2} & B_{2} \\ 
0 & 0 & H_{L_{m_{2}}^{T}} & 0 & A_{3} & B_{3} \\ 
0 & 0 & 0 & H_{L_{m_{1}}^{T}} & A_{4} & B_{4}%
\end{array}
\right) \left( 
\begin{array}{cccc}
Id_{m_{3}}\left( \lambda-\lambda_{1}\right) & 0 & 0 & 0 \\ 
0 & Id_{m_{4}}\left( \lambda-\lambda_{2}\right) & 0 & 0 \\ 
0 & 0 & L_{m_{2}}^{T} & 0 \\ 
0 & 0 & 0 & L_{m_{1}}^{T} \\ 
0 & 0 & 0 & 0 \\ 
0 & 0 & 0 & 0%
\end{array}
\right)  \notag \\
& =\left( 
\begin{array}{cccc}
H_{Id_{m_{3}}}Id_{m_{3}}\left( \lambda_{1}\right) & 0 & 0 & 0 \\ 
0 & H_{Id_{m_{4}}}Id_{m_{4}}\left( \lambda_{2}\right) & 0 & 0 \\ 
0 & 0 & H_{L_{m_{2}}^{T}}L_{m_{2}}^{T} & 0 \\ 
0 & 0 & 0 & H_{L_{m_{1}}^{T}}L_{m_{1}}^{T}%
\end{array}
\right)   \label{Sym_HK}
\end{align}
with $H_{L_{m_{1}}^{T}},$ $H_{L_{m_{2}}^{T}}$ as in eq. (\ref{HLm}). The $%
H_{Id_{m_{3}}}$, $H_{Id_{m_{4}}}$ blocks are arbitrary positive definite
Hermitian forms, and $A_{i}$, $B_{i}$ for $i=1,2,3,4$ are arbitrary blocks.
Notice that this example satisfies conditions a) and b).

As we mention before, in Geroch' formalism, which incorporates most physical
examples there are no $L_{m}^{T}$ blocks with $m\geq2$. Only $L_{1}^{T}$ and 
$L_{0}^{T}-$rows appear. These kind of blocks allow the introduction of other  reductions linking different $L_{1}^{T}$ blocks. This particular
case will be discussed in Appendix \ref{appendix_V}.

$ii\Leftrightarrow iii$

We recall that the order of the perturbed singular values is invariant under
change of bases and choice of dot products \cite{abalos2017necessary}, so we
choose bases for which $\mathfrak{N}_{~\eta }^{Ab}l\left( \lambda \right)
_{b}=K_{~\alpha }^{B}\left( \lambda \right) $, and consider, in those bases,
the positive definite Hermitian forms $G_{AB}=\left( 1,...,1\right) $ and $%
G_{\alpha \gamma }=\left( 1,...,1\right) $ to define the adjoint operator.
With this choice the computation of the singular values decouples into
blocks. Thus we only need to check the form of the perturbed singular values
for each of the $J_{l}\left( \lambda _{i}\right) $ and $L_{m}^{T}$ blocks.

The $l-$dimensional Jordan block has perturbed singular values of order $%
O\left( \varepsilon^{l}\right) $ (see \cite{abalos2017necessary}) therefore
the perturbed singular values are order $O\left( \varepsilon^{1}\right) $ if
and only if the Jordan blocks are $1-$dimensional.

On the other hand, the $L_{m}^{T}-$blocks have not right kernel for any $%
\lambda $, therefore their perturbed singular values are order $O\left(
\varepsilon ^{0}\right) $ (non vanishing). Indeed, using the relation
between the determinant relation in the SVD, we get, 
\begin{align*}
\det \left( \left( (L_{m}^{T})^{\ast }\right) _{~i}^{s}\left(
L_{m}^{T}\right) _{~j}^{i}\right) & =\sigma _{1}^{2}\left[ L_{m}^{T}\right]
...\sigma _{m}^{2}\left[ L_{m}^{T}\right]  \\
& =\left( \left\vert \lambda \right\vert ^{2}\right) ^{m}+\left( \left\vert
\lambda \right\vert ^{2}\right) ^{m-1}...+\left\vert \lambda \right\vert
^{2}+1>0
\end{align*}%
where the $\sigma _{j}\left[ L_{m}^{T}\right] $ are the singular values of $%
L_{m}^{T}$. Thus they never vanish, and when perturbed there is always a
neighborhood in which they remain positive. Thus they are order $O\left(
\varepsilon ^{0}\right) $.

As an example%
\begin{equation*}
\det\left( \left( 
\begin{array}{ccc}
\bar{\lambda} & 1 & 0 \\ 
0 & \bar{\lambda} & 1%
\end{array}
\right) \left( 
\begin{array}{cc}
\lambda & 0 \\ 
1 & \lambda \\ 
0 & 1%
\end{array}
\right) \right) =\left( \left\vert \lambda\right\vert ^{2}\right)
^{2}+\left( \left\vert \lambda\right\vert ^{2}\right) +1 
\end{equation*}

$i\Rightarrow iii)$

This implication was establish in \cite{abalos2017necessary}. This completes
the proof of the Lemma.
\end{proof}

\subsection{Choosing extra eigenvalues\label{theo_well_posed}}

In the proof of the Lemma above we constructed families of reductions that
make the system Hermitizable. That is, reductions $h_{~C}^{\rho }$ such that 
\begin{equation}
A_{~\eta }^{\gamma b}k_{b}=\left( \left( h_{~C}^{\varepsilon }\mathfrak{N}%
_{~\delta }^{Ca}n_{a}\right) ^{-1}\right) _{~\alpha }^{\gamma
}h_{~B}^{\alpha }\mathfrak{N}_{~\eta }^{Bb}k_{b}  \label{P_matrix}
\end{equation}%
is diagonalizable with only real eigenvalues for all $k_{a}$ not
proportional to $n_{a}$.

Notice that if $h_{~C}^{\rho }$ is one of these reductions, then so is  \newline $%
\tilde{h}_{~B}^{\gamma }=\left( \left( h_{~C}^{\varepsilon }\mathfrak{N}%
_{~\delta }^{Ca}n_{a}\right) ^{-1}\right) _{~\alpha }^{\gamma
}h_{~B}^{\alpha }$, since it gives the same matrix (\ref{P_matrix}),   \newline (here $%
\tilde{h}_{~B}^{\gamma }\mathfrak{N}_{~\alpha }^{Ba}n_{a}=\delta _{~\alpha
}^{\gamma }$ is the identity matrix). When $\tilde{h}_{~B}^{\gamma }$ is
written in term of the Kronecker decomposition (\ref{MR_Kro}) it assumes a
simpler form, 
\begin{equation}
\tilde{h}_{~B}^{\gamma }=-\left( W^{-1}\right) _{~\rho }^{\gamma }\left(
\left( H_{\delta C}I_{~\gamma }^{C}\right) ^{-1}\right) ^{\rho \delta
}H_{\delta C}\left( Y^{-1}\right) _{~B}^{C}.  \label{reduction_2}
\end{equation}%
This $\tilde{h}_{~B}^{\gamma }$ does not depend on $S$ matrix, showing that
this degrees of freedom does not play any role in the reductions. We shall
use it in what follows. Recall that when a reduction is applied to the
system the kernel of the resulting operator will increase, there will be $m$
more elements from each $L_{m}^{T}$ block. We shall denote the values of $%
\lambda $ for which the kernels appear as $\{\pi _{i}\left( k\right) \}$
with $i\in E_{\left( k\right) }:=\left\{ 1,2,..,u-{\displaystyle%
\sum\limits_{j\in D_{\left( k\right) }}}r_{j}\left( k\right) \right\} $.
Notice also that since the $L^{T}$ blocks can change from point to point,
(in $k_{a}$), the new eigenvalues can also change, nevertheless the
diagonalizability of $A$ implies that at all points the number of the
generalized eigenvalues plus these new elements equals the dimension of the
field space.
\begin{lemma}
\label{lemma_2} Assume condition $iii$ of lemma \ref{main_lema} holds, then
there exists $\tilde{h}_{~B}^{\gamma}$ as eq. (\ref{reduction_2}) such that:

All the $\pi _{i}\left( k\right) $ are different among each other and also
different from the generalized eigenvalues $\lambda _{i}\left( k\right) $,
with $i\in D_{\left( k\right) }$, and have algebraic multiplicity equal to $1
$.
\end{lemma}
\begin{proof}
Consider the form of $\tilde{h}_{~B}^{\gamma}$\ in equation (\ref{reduction_2})%
. The corresponding reduction block of the $L_{m}^{T}$ block is given by%
\begin{equation*}
\tilde{H}_{L_{m}^{T}}=\left( 
\begin{array}{ccccc}
1 & 0 & ... & 0 & \tilde{a}_{m+1} \\ 
0 & 1 & ... & 0 & \tilde{a}_{m+2} \\ 
... & ... & ... & 0 & ... \\ 
0 & 0 & 0 & 1 & \tilde{a}_{2m}%
\end{array}
\right) 
\end{equation*}
with 
\begin{equation*}
\left( 
\begin{array}{c}
\tilde{a}_{m+1} \\ 
\tilde{a}_{m+2} \\ 
... \\ 
\tilde{a}_{2m}%
\end{array}
\right) =g_{m}^{-1}\left( 
\begin{array}{c}
a_{m+1} \\ 
a_{m+2} \\ 
... \\ 
a_{2m}%
\end{array}
\right) 
\end{equation*}
where $g_{m}$ and $a_{i}$ where defined in equation (\ref{g_l}). Among all
possible reductions we shall now look for very special ones, namely those
for which all eigenvalues of the block are different. To find them we just
need to give values for some of the generic coefficients $\tilde{a}_i$.
Notice that if we find a reduction for which the block has different
eigenvalues, then the block will be diagonalizable, and so there will be
coefficients $a_i$ satisfying the positivity condition required. But for the
rest of the construction we shall not need to find them.

Indeed, if $\mathfrak{N}_{~\eta }^{Cb}l\left( \lambda \right) _{b}$ has an $%
L_{m}^{T}$ block in its Kronecker decomposition, then, 
\begin{align*}
\det \left( \tilde{h}_{~C}^{\rho }\mathfrak{N}_{~\eta }^{Cb}l\left( \lambda
\right) _{b}\right) & \propto \det \left( \tilde{H}_{L_{m}^{T}}L_{m}^{T}%
\right)  \\
& =\det \left( \left( 
\begin{array}{ccccc}
1 & 0 & ... & 0 & \tilde{a}_{m+1} \\ 
0 & 1 & ... & 0 & \tilde{a}_{m+2} \\ 
... & ... & ... & 0 & ... \\ 
0 & 0 & 0 & 1 & \tilde{a}_{2m}%
\end{array}%
\right) \left( 
\begin{array}{cccc}
\lambda  & 0 & 0 & 0 \\ 
1 & \lambda  & 0 & 0 \\ 
0 & 1 & ... & 0 \\ 
0 & 0 & ... & \lambda  \\ 
0 & 0 & 0 & 1%
\end{array}%
\right) \right)  \\
& =\det \left( 
\begin{array}{cccc}
\lambda  & 0 & 0 & \tilde{a}_{m+1} \\ 
1 & ... & 0 & \tilde{a}_{m+2} \\ 
0 & ... & \lambda  & ... \\ 
0 & 0 & 1 & \lambda +\tilde{a}_{2m}%
\end{array}%
\right)  \\
& =\lambda ^{m}+\tilde{a}_{2m}\lambda ^{m-1}-\tilde{a}_{2m-1}\lambda
^{m-2}+... \\ & \ \ \ -\tilde{a}_{m+3}\lambda ^{2} +\tilde{a}_{m+2}\lambda -\tilde{a}%
_{m+1}
\end{align*}%
Given any set of $m$ different real numbers it is easy to choose the
coefficients $\tilde{a}_{i}$ so that the polynomial has them as roots. This
will fix the desired reduction.
\end{proof}


\subsection{Uniform lower bound and strong hyperbolicity\label{theo_well_posed_II}}

Finally, we the help of the particular reduction we have constructed (eq. (%
\ref{reduction_2})), we shall use theorem \ref{Theorem_1} to conclude the
proof. Notice that the reductions $\tilde{h}_{~C}^{\rho}$ depends on the
bases used when doing the Kronecker decomposition, since it depends on $W$
and $Y$. These bases are not unique, and we shall use this freedom to choose
an appropriate reduction $\tilde{h}_{~C}^{\rho}$.

The reduced system is%
\begin{equation}
\tilde{h}_{~C}^{\rho }\mathfrak{N}_{~\eta }^{Cb}l\left( \lambda \right)
_{b}=-\lambda \delta _{~\eta }^{\rho }+A_{~\eta }^{\rho a}k_{a}=-\lambda
\delta _{~\eta }^{\rho }-\left( W^{-1}\right) _{~\gamma }^{\rho }\left(
\left( H_{\delta C}I_{~\gamma }^{C}\right) ^{-1}\right) ^{\gamma \delta
}H_{\delta C}M_{~\alpha }^{C}W_{~\eta }^{\alpha },  \label{sys_red_1}
\end{equation}%
where we still have freedom in choosing for our advantage $W$.

It has kernel when $\lambda =\pi _{i}\left( k\right) $ and $\lambda =\lambda
_{i}\left( k\right) $. In order to apply theorem \ref{Theorem_1}, we need to
calculate the cosines of the angles $\cos \theta ^{\pi _{i}\left( k\right) }$%
, $\cos \theta _{j}^{\lambda _{i}\left( k\right) }$ between its kernels, $%
\Upsilon _{L}^{\pi _{i}\left( k\right) }$, $\Phi _{R}^{\pi _{i}\left(
k\right) }$ and $\Upsilon _{L}^{\lambda _{i}\left( k\right) }$, $\Phi
_{R}^{\lambda _{i}\left( k\right) }$ respectively and show that they are
uniformly bounded by below.

Since each $\pi _{i}\left( k\right) $ is a simple eigenvalues of $A_{~\eta
}^{\rho a}k_{a}$, using the implicit function theorem it is possible to show
that $\cos \theta ^{\pi _{i}\left( k\right) }$ is continuous in $k_{a}$ and
since $k_{a}$ belongs to a compact set, then $\cos \theta ^{\pi _{i}\left(
k\right) }$ reaches its lower bound on that set. But since they are simple,
their cosines can not vanish for any $j$, (when perturbed, the corresponding
vanishing singular values must be of order 1) therefore the lower bound has
to be positive. Notice that for this conclusion we do not need any
information about the $W$ transformation.

To finish the proof we only need to calculate $\cos \theta _{j}^{\lambda
_{i}\left( k\right) }$ for each $\lambda _{i}(k)$. Notice that if any of the $
\lambda _{i}$ are simple, then we can use the above argument, so the
interesting case is when we have non-trivial blocks. Given any one of them,
since the $\pi _{i}\neq \lambda _{j}$, t the right kernel subspace $\Phi
_{R}^{\lambda _{i}\left( k\right) }$  is invariant under
the application of the reduction of the corresponding Jordan-block. But
notice that the left kernel of (\ref{sys_red_1}), $\Upsilon _{L}^{\lambda
_{i}\left( k\right) }$, depends on $W$. We now need to accommodate $W$ so
that the angles we are looking for coincide with the angles of the unreduced
system. For that we look now for the left kernel of the whole system, $\Psi
_{L}^{\lambda _{i}}$. This kernel has an invariant subspace whose dimension
is independent of $\lambda $, we call it $\Delta (\lambda )$. This subspace
is uniquely defined in the Kronecker decomposition of a the principal
symbol, and it is the span of a set of particular vectors which are linear
combinations with coefficients which are power laws in $\lambda $, they are
introduced in Appendix \ref{appendix}. The kernel increases its dimension by 
$r_{i}$ for each specific $\lambda _{i}$. The $\Delta (\lambda )$ subspace
has the important property that when projected with $\mathfrak{N}_{~\eta
}^{Ab}n_{b}G^{\eta \gamma }$ is orthogonal to the right kernel $\Phi
_{R}^{\lambda _{i}\left( k\right) }$ (see Appendix \ref{app_IIIb}). So we
now project into $\Phi _{R}$ the whole kernel with $\mathfrak{N}_{~\eta
}^{Ab}n_{b}G^{\eta \gamma }$ and call the resulting subspace $\Phi
_{L}^{\lambda _{i}}$. The projected image of $\Delta (\lambda )$ will be
called $\Delta _{\mathfrak{N}}(\lambda )$. Using the metric $G_{\alpha
\gamma }$ we write $\Phi _{L}^{\lambda _{i}}$ as the direct sum of $\Delta _{%
\mathfrak{N}}(\lambda_i )$ and its perpendicular inside it, $(\Phi _{L}^{\lambda_i})^{\perp }$. Since $\Delta _{\mathfrak{N}}(\lambda_i )$ is perpendicular to $\Phi
_{R}^{\lambda _{i}\left( k\right) }$ and to $(\Phi _{L}^{\lambda_i})^{\perp }$, the angles
between $(\Phi _{L}^{\lambda_i})^{\perp }$ and $\Phi _{R}^{\lambda _{i}}$ are the same as
the angles between $\Phi _{L}^{\lambda _{i}}$ and $\Phi _{R}^{\lambda _{i}}$%
, which are the ones that appear as our theorem's hypothesis, and so their
cosines are bounded away from zero. We want to find now a $W$ such that $%
\Upsilon _{L}^{\lambda _{i}}$ coincides with $(\Phi _{L}^{\lambda_i})^{\perp }$, and so
their respective angles. To do that we choose a set of $r_{i}$ linearly
independent vectors $\{v_{A}^{l}\}$ in $\Psi _{L}^{\lambda _{i}}$ such that $%
span\{v_{A}^{l}\mathfrak{N}_{~\eta }^{Ab}n_{b}G^{\eta \gamma }\}=(\Phi _{L}^{\lambda_i})^{\perp }$. Taking now the set of canonical $e-u$ vectors in $\Delta
(\lambda_i )$, $\{\chi _{A}^{i}\}$ we obtain a base for $\Psi _{L}^{\lambda
_{i}}$, $\{\chi _{A}^{i},v_{A}^{l}\}$ (this base defines $W$, see appendix %
\ref{appendix}). Indeed assume that the above vectors are not linearly
independent, that is, we can write a vector in $\Delta (\lambda_i )$ as a
linear combination of the other vectors, $\chi _{A}=a_{l}v_{A}^{l}$.
Contracting with $\mathfrak{N}_{~\eta }^{Ab}n_{b}G^{\eta \gamma }$ we get, $%
\chi _{A}\mathfrak{N}_{~\eta }^{Ab}n_{b}G^{\eta \gamma }=a_{l}v_{A}^{l}%
\mathfrak{N}_{~\eta }^{Ab}n_{b}G^{\eta \gamma }$. Now, the LHS is an element
of $\Delta _{\mathfrak{N}}(\lambda_i )$ while the RHS is an element of $(\Phi _{L}^{\lambda_i})^{\perp }$. Since these spaces are perpendicular to each other we
conclude that both must vanish. But since by assumption the $v_{A}^{l}%
\mathfrak{N}_{~\eta }^{Ab}n_{b}G^{\eta \gamma }$ are linearly independent we
conclude that the $a_{l}$ must vanish, and so reach a contradiction. Using
this base it is straightforward to see that the resulting $W$ has the
property that $(\Phi _{L}^{\lambda_i})^{\perp }$ is the left kernel of the reduced system,
thus it coincides with $\Upsilon _{L}^{\lambda _{i}}$ and have the same
angles. This concludes the proof.

We conclude the subsection showing how to apply make this construction in
practical examples. Given a principal symbol we perform its Kronecker
decomposition. This provides us with some bases, 
\begin{equation}
\chi_{A}^{s}\left( \lambda\right) \text{ and }\left(
\upsilon_{\lambda_{i}}^{l}\right) _{A}   \label{cov_1}
\end{equation}
where $\chi_{A}^{s}\left( \lambda\right) :=\left( \left(
\theta_{m_{s}}^{s}\right) _{A}-\lambda\left( \theta_{m_{s}-1}^{s}\right)
_{A}+\lambda^{2}\left( \theta_{m_{s}-2}^{s}\right)
_{A}-...-\lambda^{m_{1}}\theta_{0A}^{s}\right) $ are a set of vectors which
are in the left kernel for all $\lambda$, While the rest of the base vectors
are only kernels for particular values of $\lambda$. The subspace spanned by 
$\{\chi_{A}^{s}\left( \lambda\right) \}$ is what we called $\Delta(\lambda)$
above. For each $\lambda_i$ the span of the whole set of these vectors, $%
\{\chi_{A}^{s}\left(
\lambda_i\right),\left(\upsilon_{\lambda_{i}}^{l}\right)_{A}\}$ span the left
kernel, $\Psi^{\lambda_i}_L$. But this base is not unique, we shall change
the elements $\{\left( \upsilon_{\lambda_{i}}^{l}\right)_{A}\}$ to find
simple reductions $h_{~A}^{\alpha}$. From equation (\ref{reduction_2}) we
see that once we fix the $\{\left( \upsilon_{\lambda_{i}}^{l}\right)_{A}\}$
the reduction if also fixed, for $Y$ also depends only on these vectors. As
mention above for each $i \in I_{\lambda _{i}\left( k\right)} $ we now choose $r_{\lambda_i(k)}$ new vectors $\{\left( \tilde{\upsilon}%
_{\lambda_{i}}^{l}\right)_{A}\}$ so that when they are projected into $\Phi_R$, they span the subspace $(\Phi _{L}^{\lambda_i(k)})^{\perp }$.

\section{\label{sec:level1:4} Examples}
\subsection{Klein Gordon}
In this subsection we study the Klein Gordon equation in \newline Minkowski space-time. 
We show that the Kronecker decomposition of the principal symbol is $1\times
J_{1}(1)$, $1\times
J_{1}(-1)$, $3\times L_{1}^{T}$ and $3\times L_{0}^{T}$, with generalized
eigenvalues $\pm 1$. As it is shown in \cite{geroch1996partial}, this system
is symmetric hyperbolic, hence it is strongly hyperbolic. We shall show the
possible reductions of the systems. 

The Klein Gordon equation is%
\begin{equation*}
g^{ab}\nabla_{b}\nabla_{a}\phi=0. 
\end{equation*}
This equation can be written in first order form, introducing new variables, 
\begin{equation}
\phi_{a}:=\nabla_{a}\phi.   \label{KG_2}
\end{equation}

We obtain eleven equations for five variables, $\left( \phi,\phi_{a}\right)$%
. They are 
\begin{align*}
\nabla_{b}\phi - \phi_{b} & =0 \\
\nabla^{b}\phi_{b} & = 0 \\
\nabla_{\lbrack a}\phi_{b]} &=0
\end{align*}

Taking the Fourier transform we obtain the principal symbol, 
\begin{equation*}
\left( 
\begin{array}{cc}
\delta _{a}^{d} & 0 \\ 
0 & g^{dl} \\ 
0 & \delta _{\lbrack b}^{d}\delta _{c]}^{l}%
\end{array}%
\right) l\left( \lambda \right) _{d}\left( 
\begin{array}{c}
\delta \phi  \\ 
\delta \phi _{l}%
\end{array}%
\right) =0
\end{equation*}

Choosing a time-like co-vector $n_{a}$ and lines $l_{a}\left( \lambda
\right) =-\lambda n_{a}+k_{a}\in S_{n_{a}}^{\mathbb{C}}$, we obtain the
matrix pencil form of the principal part,

\begin{equation*}
-\lambda \left( 
\begin{array}{cc}
\delta _{a}^{d} & 0 \\ 
0 & g^{dl} \\ 
0 & \delta _{\lbrack b}^{d}\delta _{c]}^{l}%
\end{array}%
\right) n_{d}+\left( 
\begin{array}{cc}
\delta _{a}^{d} & 0 \\ 
0 & g^{dl} \\ 
0 & \delta _{\lbrack b}^{d}\delta _{c]}^{l}%
\end{array}%
\right) k_{d}
\end{equation*}

we are considering $n.n=-1$, $k.k=1$, and $n.k=0$.

Following the appendix \ref{appendix_II} the left kernel is spanned by the co-vectors 
$\{\left( \hat{\theta}\right) _{0A},$ $\left( \tilde{\theta}_{i}\right)
_{0A},\\ -\lambda \theta _{0A}+\theta _{1A},-\lambda \left( \breve{\theta}%
_{i}\right) _{0A}+\left( \breve{\theta}_{i}\right) _{1A}\}$ with $i=1,2$,
which span the subspace $\Delta (\lambda )$, and the eigen-covectors $%
\{\upsilon _{1A},\upsilon _{2A}\}$ associated to the generalized eigenvalues 
$\lambda _{\pm }=\pm 1$. The Kronecker left base is then, 
\begin{align*}
\left( \hat{\theta}\right) _{0A}& =\left( 
\begin{array}{ccc}
0 & 0 & \varepsilon ^{k_{1}bcd}n_{d}k_{k_{1}}%
\end{array}%
\right) \rightarrow 1 \\
\left( \tilde{\theta}_{i}\right) _{0A}& =\left( 
\begin{array}{ccc}
\varepsilon ^{k_{1}ade}n_{d}\left( l_{i}\right) _{e}k_{k_{1}} & 0 & 0%
\end{array}%
\right) \rightarrow 2 \\
\theta _{0A}& =\left( 
\begin{array}{ccc}
\varepsilon ^{k_{1}ade}l_{1d}l_{2e}n_{k_{1}} & 0 & 0%
\end{array}%
\right) \rightarrow 1 \\
\theta _{1A}& =\left( 
\begin{array}{ccc}
\varepsilon ^{k_{1}ade}l_{1d}l_{2e}k_{k_{1}} & 0 & 0%
\end{array}%
\right) \rightarrow 1 \\
\left( \breve{\theta}_{i}\right) _{0A}& =\left( 
\begin{array}{ccc}
0 & 0 & \varepsilon ^{k_{1}kla}\left( l_{i}\right) _{a}n_{k_{1}}%
\end{array}%
\right) \rightarrow 2 \\
\left( \breve{\theta}_{i}\right) _{1A}& =\left( 
\begin{array}{ccc}
0 & 0 & \varepsilon ^{k_{1}kla}\left( l_{i}\right) _{a}k_{k_{1}}%
\end{array}%
\right) \rightarrow 2 \\
\upsilon _{iA}& =\left( 
\begin{array}{ccc}
0 & -\frac{1}{2}\lambda _{i} & n^{[b}k^{c]}%
\end{array}%
\right) \rightarrow 2
\end{align*}

where $l_{i}.k=l_{i}.n=0$ and $l_{i}.l_{j}=\delta _{ij}$ with $i,j=1,2$

With this set we build the Kronecker decomposition as in equation (\ref%
{MR_Kro})%
\begin{eqnarray*}
&&Y_{~B}^{A}\left( n,k\right)  \\
&=&\left( 
\begin{array}{ccccccccccc}
0 & 0 & 0 & 0 & 0 & 0 & n_{a} & -k_{a} & -l_{1a} & l_{2a} & 0 \\ 
-1 & 1 & 0 & 0 & 0 & 0 & 0 & 0 & 0 & 0 & 0 \\ 
-n_{[b}k_{c]} & -n_{[b}k_{c]} & n_{[b}l_{1c]} & -k_{[b}l_{1c]} & 
-n_{[b}l_{2c]} & k_{[b}l_{2c]} & 0 & 0 & 0 & 0 & l_{1[b}l_{2c]}%
\end{array}%
\right) 
\end{eqnarray*}

\begin{equation*}
K_{~\alpha }^{B}\left( \lambda \right) =\left( 
\begin{array}{ccccc}
\lambda -1 & 0 & 0 & 0 & 0 \\ 
0 & \lambda +1 & 0 & 0 & 0 \\ 
0 & 0 & \lambda  & 0 & 0 \\ 
0 & 0 & 1 & 0 & 0 \\ 
0 & 0 & 0 & \lambda  & 0 \\ 
0 & 0 & 0 & 1 & 0 \\ 
0 & 0 & 0 & 0 & \lambda  \\ 
0 & 0 & 0 & 0 & 1 \\ 
0 & 0 & 0 & 0 & 0 \\ 
0 & 0 & 0 & 0 & 0 \\ 
0 & 0 & 0 & 0 & 0%
\end{array}%
\right) \text{ \ \ \ \ \ \ \ \ }W_{~\eta }^{\alpha }\left( n,k\right)
=\left( 
\begin{array}{cc}
0 & \frac{1}{2}\left( n^{l}+k^{l}\right)  \\ 
0 & \frac{1}{2}\left( -n^{l}+k^{l}\right)  \\ 
0 & -l_{1}^{l} \\ 
0 & l_{2}^{l} \\ 
-1 & 0%
\end{array}%
\right) 
\end{equation*}%

Following eq. (\ref{reduction_2}) the reductions are: 
\begin{equation*}
\tilde{h}=W^{-1}\left( 
\begin{array}{ccccccccccc}
1 & 0 & 0 & 0 & 0 & 0 & 0 & 0 & e_{1} & d_{1} & f_{1} \\ 
0 & 1 & 0 & 0 & 0 & 0 & 0 & 0 & e_{2} & d_{2} & f_{2} \\ 
0 & 0 & 1 & a_{1} & 0 & a_{2} & 0 & a_{3} & e_{3} & d_{3} & f_{3} \\ 
0 & 0 & 0 & \bar{a}_{2} & 1 & b_{1} & 0 & b_{2} & e_{4} & d_{4} & f_{4} \\ 
0 & 0 & 0 & \bar{a}_{3} & 0 & \bar{b}_{2} & 1 & c_{3} & e_{5} & d_{5} & f_{5}%
\end{array}%
\right) Y^{-1}
\end{equation*}%
where the coefficients in $\tilde{h}_{~B}^{\gamma }$ are arbitrary complex function of $n$
and $k$, with the exception of $a_{1}$, $b_{1}$ and $c_{3}$ which are real.

The pseudo-differential evolution equations (principal part) are 
\begin{equation*}
\left( 
\begin{array}{c}
\partial _{t}\hat{\phi} \\ 
\partial _{t}\hat{\phi}_{l}%
\end{array}%
\right) =-\left( 
\begin{array}{cc}
-c_{3} & k^{b}n^{a}\bar{R}^{c}\varepsilon _{~bac}^{m} \\ 
k^{b}n^{a}R^{c}\varepsilon _{~lbac} & -k_{l}n^{m}+n_{l}k^{m}-ia_{2I}%
\varepsilon _{l}^{~bma}n_{b}k_{a}+S_{~l}^{m}%
\end{array}%
\right) \left( 
\begin{array}{c}
\hat{\phi} \\ 
\hat{\phi}_{m}%
\end{array}%
\right) 
\end{equation*}%
with $a_{2}=a_{2R}+ia_{2I}$, $S_{~l}^{m}=-a_{2R}\left(
l_{1l}l_{2}^{m}-l_{2l}l_{1}^{m}\right)
+a_{1}l_{1l}l_{1}^{m}+b_{1}l_{2l}l_{2}^{m}$, and $R^{c}$ is any complex
vector.

It is instructive to look now at the possible differential reductions. In
Cartesian adapted coordinates the Klein Gordon system becomes, 
\begin{align*}
\partial_{t}\phi & =\phi_{0} \\
\partial_{t}\phi_{0} & =-\partial^{i}\phi_{i} \\
\partial_{t}\phi_{i} & =\partial_{i}\phi_{0} \\
C_{i} & :=\partial_{i}\phi-\phi_{i}=0 \\
C_{ij} & :=\partial_{\lbrack i}\phi_{j]}=0,
\end{align*}
Where the last two equations are clearly the constraints.

The most general differential hyperbolization is obtained by setting $%
S_{~l}^{m}=0$ and $c_{3}=L^{i}k_{i}$. The equations for the principal part
becomes,

\begin{align*}
\partial_{t}\phi &  =L^{i}\partial_{i}\phi-\bar{R}_{k}\varepsilon
^{ijk}\partial_{i}\phi_{j}=L^{i}C_{i}-\bar{R}_{k}\varepsilon^{ijk}C_{ij}\\
\partial_{t}\phi_{0} &  =-\partial^{i}\phi_{i}\\
\partial_{t}\phi_{i} &  =\partial_{i}\phi_{0}+R_{k}\varepsilon_{i}%
^{~jk}\partial_{j}\phi-ia_{2I}\varepsilon_{l}^{~mj}\partial_{j}\phi
_{m}=\partial_{i}\phi_{0}+R_{k}\varepsilon_{i}^{~jk}C_{j}+ia_{2I}%
\varepsilon_{l}^{~ij}C_{ij}%
\end{align*}
This expression clearly shows that the freedom in the choice of reductions is the addition arbitrary linear combination of constraints to some of the equations.



\subsection{Einstein equations.}

In the following two sections we study  in a
pseudo-differential form the
linearized, densitized ADM and a version of the BSSN equations of General Relativity. 
We present their Kronecker decompositions and  show hyperbolizers for them. 
 We do not explain the details of how to arrive to the pseudo-differential equations, the complete 
theory can be found in the  work of Nagy, Reula and Ortiz
\cite{nagy2004strongly}. 
We show that the
principal symbol of these system have
the following  Kronecker structure, both share a  diagonal part  of $2 \times J_{1}(0)$, $2 \times J_{1}(1)$, 
$2 \times J_{1}(-1)$, $1 \times J_{1}(\sqrt{b})$ and $1 \times J_{1}(-\sqrt{b})$. 
This $8\times8$ block  corresponds to the physical propagation. In addition, the ADM equations have  $4\times L_{1}^{T}$ blocks and the BSSN equations have  $7\times L_{1}^{T}$  and  $6\times L_{0}^{T}$ blocks. 
Since the systems have not any $J_{m}$ Jordan block with $m \geq 2$, the Lemma \ref{main_lema}
implies that there exists diagonalizable reductions of the systems. 
We will present some of them, check the uniformity condition and so conclude that they are proper hyperbolizations. 

In addition, we  show that in the non densitized BSSN theory ($b=0$), the Kronecker decomposition have a
Jordan block of dimension $2$, and therefore that there exists no hyperbolizers for it. 
This is because the $1 \times J_{1}(\sqrt{b})$ and $1 \times J_{1}(-\sqrt{b})$ blocks collapse to a $1 \times J_{2}(0)$. 
The same calculation can be made in the non densitized ADM equations, with  similar conclusions.




\subsubsection{ADM}

The 3+1 principal part of the ADM equations, including the constraints
equations, for the linearizeds field $\left(  h_{ij},k_{ij}\right)  $on a flat background are:

\begin{align*}
\frac{1}{N}\partial_{t}h_{ij}  &  =-2k_{ij}\\
\frac{1}{N}\partial_{t}k_{ij}  &  =\frac{1}{2}e^{kl}\left(  -\partial
_{k}\partial_{l}h_{ij}-\left(  1+b\right)  \partial_{i}\partial_{j}%
h_{kl}+2\partial_{k}\partial_{(i}h_{j)l}\right) \\
0  &  =\left(  -\Delta\bar{h}+\partial^{l}\partial^{j}h_{jl}\right) \\
0  &  =\partial^{l}k_{li}-\partial_{i}\bar{k}%
\end{align*}

Here $e_{ij}$ is the hypersuface-induced flat background metric and $\partial_i$ its connection. 
The lapse is densitized  as $N=h^{\frac{1}{2}b}Q$ with $Q$ fixed, and the shift vector is taken to be zero. 
The lower order terms can be found in \cite{nagy2004strongly}.

Taking a Fourier transform $\left(  h_{ij},k_{ij}\right)  \rightarrow
\left(  \hat{h}_{ij},\hat{k}_{ij}\right)  e^{i\left(
-\tilde{\lambda}t+k_{i}x^{i}\right)  }$ and defining $\hat{l}%
_{ij}=i\left\vert k\right\vert _{e}\hat{h}_{ij}$ with $\left\vert
k\right\vert _{e}=\sqrt{e^{ij}k_{i}k_{j}}$ as in
\cite{nagy2004strongly} we obtain the following principal symbol. \
\begin{align*}
\mathfrak{N}_{~\eta}^{Ab}n\left(  \lambda\right)  _{b}\delta\phi^{\alpha}  &
=\left(  \lambda E_{\text{ }\alpha}^{A~~}+B_{~\alpha}^{A~~}\right)  \left(
\begin{array}
[c]{c}%
\hat{l}_{sr}\\
\hat{k}_{sr}\\
\hat{f}_{s}%
\end{array}
\right) \\
&  =\left(
\begin{array}
[c]{cc}%
\lambda\delta_{(i}^{s}\delta_{j)}^{r} & 2\delta_{(i}^{s}\delta_{j)}^{r}\\
\frac{1}{2}\left(  \delta_{(i}^{s}\delta_{j)}^{r}+\left(  1+b\right)
\tilde{k}_{i}\tilde{k}_{j}e^{sr}-2\tilde{k}_{(i}\tilde{k
}^{(s}\delta_{j)}^{r)}\right)  & \lambda\delta_{(i}^{s}\delta_{j)}^{r}\\
-\alpha q^{sr} & 0\\
0 & \left(  \tilde{k}^{l}\delta_{(i}^{s}\delta_{l)}^{r}-e^{sr}\tilde
{k}_{i}\right)
\end{array}
\right)  \left(
\begin{array}
[c]{c}%
\hat{l}_{sr}\\
\hat{k}_{sr}%
\end{array}
\right)  =0
\end{align*}
for the variables $\hat{l}_{sr}, 
\hat{k}_{sr}$.

Notice that we are considered the lines $-\lambda n_a + k_i \in S_{n_{a}}$ with $n_{a}%
=\nabla_{a}t$ and $t$ the time coordinate,  also $\lambda:=\frac
{\tilde{\lambda}}{N\left\vert k\right\vert _{e}}$ and $q^{sr}%
:=e^{sr}-\tilde{k}^{s}\tilde{k}^{r}$ is the proyector orthogonal to $\tilde{k}^{r}.$ 
We denote $\tilde{k}_{i}$ to the the covector such that $\left\vert \tilde{k}\right\vert _{e}=1$. 
For details about the pseudo-differential structure see \cite{nagy2004strongly} (a different notation is using here with respect to that article,  $\omega_{i}\rightarrow k_{i}$ and $\alpha=1$). The characteristic speeds are solutions to (\ref{p_1}), which becomes, 

\begin{align*}
&  \sqrt{\det\left(  G^{\gamma\alpha}\overline{\left(  \lambda E_{\text{
}\alpha}^{A}+B_{~\alpha}^{A}\right)  }G_{AB}\left(  \left(  \lambda
E_{\text{~}\rho}^{B}+B_{~\rho}^{B}\right)  \right)  \right)  }\\
&  = \left\vert \lambda\right\vert^{2} \left\vert \left(  \lambda
^{2}-1\right)  \right\vert^{2} \left\vert \left(  b-\lambda^{2}\right)
\right\vert \sqrt{p_{1}\left(  \lambda,\bar{\lambda}\right)  }=0%
\end{align*}
Here $G^{\gamma \alpha}$ and  $G_{AB}$ are any positive definite Hermitian forms, and $p_{1}\left(  \lambda,\bar{\lambda}\right)  $ a positive polynomial in $\lambda$, $\bar{\lambda}.$ 
The bar means conjugation. Therefore the mode propagation velocities, or equivalently the
generalized eigenvalues of the complete system  are: $\lambda=0,$ $\lambda=\pm1$ and $\lambda=\pm\sqrt{b}.$

The left kernel of $M_{~\alpha}^{A}$ are vectors to the form $\delta
X_{A}=\left(
\begin{array}
[c]{cccc}%
\alpha^{ij}, & \gamma^{ij}, & \rho, & \tau^{i}%
\end{array}
\right)  $ 
It is useful to define the vectors $x_{\pm}$ with  $x_{\pm}.x_{\pm}=\delta_{\pm}\,\ $ (the Kronecker delta) and $\tilde{k
}.x_{\pm}=0.$
A convenient base for the generalized left kernel is:

For the generalized eigenvalue $\lambda=0$, we have two eigenvectors%
\[
\upsilon_{1\pm A}=\left(
\begin{array}
[c]{cccc}%
2\tilde{k}^{(i}x_{\pm}^{j)}, & 0, & 0, & -4x_{\pm}^{i}%
\end{array}
\right) 
\]

For  $\lambda=\pm 1$,%
\begin{align*}
\upsilon_{2\pm A}  &  =\left(
\begin{array}
[c]{cccc}%
\mp\frac{1}{2}\left(  x_{+}^{i}x_{+}^{j}-x_{-}^{i}x_{-}^{j}\right),  & \left(
x_{+}^{i}x_{+}^{j}-x_{-}^{i}x_{-}^{j}\right),  & 0, & 0
\end{array}
\right) \\
\upsilon_{3\pm A}  &  =\left(
\begin{array}
[c]{cccc}%
\mp\frac{1}{2}x_{+}^{(i}x_{-}^{j)}, & x_{+}^{(i}x_{-}^{j)}, & 0, & 0
\end{array}
\right)
\end{align*}



For  $\lambda=\pm\sqrt{b}$ 
\[
\upsilon_{4\pm A}=\left(
\begin{array}
[c]{cccc}%
\mp\frac{1}{2}\sqrt{b}\tilde{k}^{i}\tilde{k}^{j}, & \tilde{k}%
^{i}\tilde{k}^{j}, & \frac{1}{2}(b+1), & 0
\end{array}
\right)
\]


The left kernel associated to the  $4 \times  L_{1}^{T}$  blocks are of order $1$ in $\lambda$:
\begin{align*}
\left(  \theta_{1}\right)  _{1A}-\lambda\left(  \theta_{0}\right)  _{1A}  &
=\left(
\begin{array}
[c]{cccc}%
\frac{1}{2}q^{ij}, & 0, & \frac{1}{2}\lambda, & \tilde{k}^{i}
\end{array}
\right) \\
\left(  \theta_{1}\right)  _{2A}-\lambda\left(  \theta_{0}\right)  _{2A}  &
=\left(
\begin{array}
[c]{cccc}%
0, & \frac{1}{2}q^{ij}, & \frac{1}{4}, & \frac{1}{2}\lambda\tilde{k}^{i}
\end{array}
\right) \\
\left(  \theta_{1}\right)  _{3\pm A}-\lambda\left(  \theta_{0}\right)  _{3\pm
A}  &  =\left(
\begin{array}
[c]{cccc}%
0, & 2\tilde{k}^{(i}x_{\pm}^{j)}, & 0, & -2\lambda x_{\pm}^{i}
\end{array}
\right)
\end{align*}

Therefore the Kronecker decomposition have the $5$ generalized eigenvalues
with their corresponding eigenvectors that defines a $8\times8$ diagonal
block,  $4 \times  L_{1}^{T}$ blocks: 
\[
K=\left(
\begin{array}
[c]{cccccccccccc}%
\lambda & 0 & 0 & 0 & 0 & 0 & 0 & 0 & 0 & 0 & 0 & 0\\
0 & \lambda & 0 & 0 & 0 & 0 & 0 & 0 & 0 & 0 & 0 & 0\\
0 & 0 & \lambda-1 & 0 & 0 & 0 & 0 & 0 & 0 & 0 & 0 & 0\\
0 & 0 & 0 & \lambda-1 & 0 & 0 & 0 & 0 & 0 & 0 & 0 & 0\\
0 & 0 & 0 & 0 & \lambda+1 & 0 & 0 & 0 & 0 & 0 & 0 & 0\\
0 & 0 & 0 & 0 & 0 & \lambda+1 & 0 & 0 & 0 & 0 & 0 & 0\\
0 & 0 & 0 & 0 & 0 & 0 & \lambda-\sqrt{b} & 0 & 0 & 0 & 0 & 0\\
0 & 0 & 0 & 0 & 0 & 0 & 0 & \lambda+\sqrt{b} & 0 & 0 & 0 & 0\\
0 & 0 & 0 & 0 & 0 & 0 & 0 & 0 & \lambda & 0 & 0 & 0\\
0 & 0 & 0 & 0 & 0 & 0 & 0 & 0 & 1 & 0 & 0 & 0\\
0 & 0 & 0 & 0 & 0 & 0 & 0 & 0 & 0 & \lambda & 0 & 0\\
0 & 0 & 0 & 0 & 0 & 0 & 0 & 0 & 0 & 1 & 0 & 0\\
0 & 0 & 0 & 0 & 0 & 0 & 0 & 0 & 0 & 0 & \lambda & 0\\
0 & 0 & 0 & 0 & 0 & 0 & 0 & 0 & 0 & 0 & 1 & 0\\
0 & 0 & 0 & 0 & 0 & 0 & 0 & 0 & 0 & 0 & 0 & \lambda\\
0 & 0 & 0 & 0 & 0 & 0 & 0 & 0 & 0 & 0 & 0 & 1
\end{array}
\right)
\]

The most general  reduction diagonal in time is: 

\[
\left(  \left(  H_{\delta C}I_{~\gamma}^{C}\right)  ^{-1}\right)  ^{\rho
\delta}H_{\delta C}=\left(
\begin{array}
[c]{cccccccccccccccc}%
1 & 0 & 0 & 0 & 0 & 0 & 0 & 0 & 0 & p_{1} & 0 & r_{1} & 0 & s_{1} & 0 &
w_{1}\\
0 & 1 & 0 & 0 & 0 & 0 & 0 & 0 & 0 & p_{2} & 0 & r_{2} & 0 & s_{2} & 0 &
w_{2}\\
0 & 0 & 1 & 0 & 0 & 0 & 0 & 0 & 0 & p_{3} & 0 & r_{3} & 0 & s_{3} & 0 &
w_{3}\\
0 & 0 & 0 & 1 & 0 & 0 & 0 & 0 & 0 & p_{4} & 0 & r_{4} & 0 & s_{4} & 0 &
w_{4}\\
0 & 0 & 0 & 0 & 1 & 0 & 0 & 0 & 0 & p_{5} & 0 & r_{5} & 0 & s_{5} & 0 &
w_{5}\\
0 & 0 & 0 & 0 & 0 & 1 & 0 & 0 & 0 & p_{6} & 0 & r_{6} & 0 & s_{6} & 0 &
w_{6}\\
0 & 0 & 0 & 0 & 0 & 0 & 1 & 0 & 0 & p_{7} & 0 & r_{7} & 0 & s_{7} & 0 &
w_{7}\\
0 & 0 & 0 & 0 & 0 & 0 & 0 & 1 & 0 & p_{8} & 0 & r_{8} & 0 & s_{8} & 0 &
w_{8}\\
0 & 0 & 0 & 0 & 0 & 0 & 0 & 0 & 1 & a_{1} & 0 & c_{1} & 0 & c_{2} & 0 &
c_{3}\\
0 & 0 & 0 & 0 & 0 & 0 & 0 & 0 & 0 & b_{1} & 1 & a_{2} & 0 & g_{1} & 0 &
g_{2}\\
0 & 0 & 0 & 0 & 0 & 0 & 0 & 0 & 0 & b_{2} & 0 & f_{1} & 1 & a_{3} & 0 &
m_{1}\\
0 & 0 & 0 & 0 & 0 & 0 & 0 & 0 & 0 & b_{3} & 0 & f_{2} & 0 & l_{1} & 1 & a_{4}%
\end{array}
\right)  
\]

It becomes the reductions considered on appendix \ref{appendix_V} when setting $p_i=r_i=s_i=w_i=0$ and  
$$A=\left(
\begin{array}
[c]{cccc}%
a_{1} & c_{1} & c_{2} & c_{3}\\
b_{1} & a_{2} & g_{1} & g_{2}\\
b_{2} & f_{1} & a_{3} & m_{1}\\
b_{3} & f_{2} & l_{1} & a_{4}%
\end{array}
\right)  $$is a diagonalizable matrix with real and simple eingenvalues, and different to the generalized eigenvalues, namely, $\{0,-1,1,\sqrt{b},-\sqrt{b}\}$.
The coefficients $a_i, b_i, c_i, g_i, f_i, m_i, l_i, p_i, r_i, s_i, w_i$ parametrize the remaining degrees of freedom of the most general reductions
$\tilde{h}_{~B}^{\gamma}=\left(  W^{-1}\right)  _{~\rho}^{\gamma}\left(
\left(  H_{\delta C}I_{~\gamma}^{C}\right)  ^{-1}\right)  ^{\rho\delta
}H_{\delta C}\left(  Y^{-1}\right)  _{~B}^{C}$. Where $Y^{-1}$ and $W$ are building using the basis above. Then, 
\[
\tilde{h}_{~B}^{\gamma}=\left(
\begin{array}
[c]{cccc}%
\delta_{(s}^{i}\delta_{r)}^{j} & 0 & A_{3sr} & A_{4sr}^{i}\\
0 & \delta_{(s}^{i}\delta_{r)}^{j} & B_{3sr} & B_{4sr}^{i}%
\end{array}
\right)
\]
\begin{align*}
& A_{3sr}=-\frac{1}{2}p_{1}\tilde{k}_{(s}x_{+r)}-\frac{1}{2}p_{2}\tilde
{k}_{(s}x_{-r)}+\left(  \frac{1}{4}p_{4}-\frac{1}{4}p_{6}\right)
x_{+(s}x_{-r)}+\left(  \frac{1}{4}p_{3}-\frac{1}{4}p_{5}-\frac{1}{2}%
a_{1}\right)  x_{+s}x_{+r}\\
& -\left(  \frac{1}{4}p_{3}-\frac{1}{4}p_{5}+\frac{1}{2}a_{1}\right)
x_{-s}x_{-r}+\left(  \frac{1}{2}p_{7}-\frac{1}{2}p_{8}\right)  \frac{1}%
{\sqrt{b}}\tilde{k}_{s}\tilde{k}_{r}%
\end{align*}
\begin{align*}
& B_{3st}=+\frac{1}{4}\left(  2b+2-p_{7}-p_{8}\right)  \tilde{k}_{s}\tilde
{k}_{r}-\frac{1}{8}\left(  p_{3}+p_{5}-2+4b_{1}\right)  x_{+s}x_{+r}\\
& +\frac{1}{8}\left(  p_{3}+p_{5}+2-4b_{1}\right)  x_{-s}x_{-r}-\frac{1}%
{8}\left(  p_{4}+p_{6}\right)  x_{+(s}x_{-r)}-\frac{1}{2}b_{2}\tilde{k}%
_{(s}x_{+r)}-\frac{1}{2}b_{3}\tilde{k}_{(s}x_{-r)}%
\end{align*}%
\begin{align*}
& A_{4sr}^{i}=+\frac{1}{\sqrt{b}}\tilde{k}_{s}\tilde{k}_{r}\left( \frac{1}%
{2}\left(  r_{7}-r_{8}\right)  \tilde{k}^{i}+2\left(  s_{8}-2s_{7}\right)
x_{+}^{i}+2\left( w_{8}-w_{7}\right)  x_{-}^{i}\right)  \\
& +\tilde{k}_{(s}x_{+r)}\left(  -\frac{1}{2}r_{1}\tilde{k}^{i}+\left(
-4+2s_{1}\right)  x_{+}^{i}+2w_{1}x_{-}^{i}\right) \\
& +\tilde{k}_{(s}%
x_{-r)}\left(  -\frac{1}{2}r_{2}\tilde{k}^{i}+2s_{2}x_{+}^{i}+\left(
-4+2w_{2}\right)  x_{-}^{i}\right)  \\
& +\frac{1}{2}x_{+s}x_{+r}\left(  \left(  -\frac{1}{2}r_{5}+\frac{1}{2}%
r_{3}+\left(  2-c_{1}\right)  \right)  \tilde{k}^{i}+\left(  2s_{5}%
-2s_{3}+4c_{2}\right)  x_{+}^{i}+\left(  2w_{5}-2w_{3}+4c_{3}\right)
x_{-}^{i}\right)  \\
& -\frac{1}{2}x_{-s}x_{-r}\left(  \left(  -\frac{1}{2}r_{5}+\frac{1}{2}%
r_{3}-\left(  2-c_{1}\right)  \right)  \tilde{k}^{i}+\left(  2s_{5}%
-2s_{3}-4c_{2}\right)  x_{+}^{i}+\left(  2w_{5}-2w_{3}-4c_{3}\right)
x_{-}^{i}\right)  \\
& +\frac{1}{2}x_{+(s}x_{-r)}\left(  \frac{1}{2}\left(  r_{4}-r_{6}\right)
\tilde{k}^{i}+2\left(  s_{6}-s_{4}\right)  x_{+}^{i}+2\left(  w_{6}%
-w_{4}\right)  x_{-}^{i}\right)
\end{align*}%
\begin{align*}
& B_{4sr}^{i}=+\frac{1}{2}\tilde{k}_{s}\tilde{k}_{r}\left(  -\frac{1}%
{2}\left(  r_{7}+r_{8}\right)  \tilde{k}^{i}+2\left(  s_{7}+s_{8}\right)
x_{+}^{i}+2\left(  w_{7}+w_{8}\right)  x_{-}^{i}\right)  \\
& +\frac{1}{4}x_{+s}x_{+r}\left(  -\frac{1}{2}\left(  4a_{2}+r_{3}%
+r_{5}\right)  \tilde{k}^{i}+2\left(  s_{3}+s_{5}+4g_{1}\right)  x_{+}%
^{i}+2\left(  w_{3}+w_{5}+4g_{2}\right)  x_{-}^{i}\right)  \\
& -\frac{1}{4}x_{-s}x_{-r}\left(  \frac{1}{2}\left(  4a_{2}-\left(
r_{3}+r_{5}\right)  \right)  \tilde{k}^{i}+2\left(  s_{3}+s_{5}-4g_{1}\right)
x_{+}^{i}+2\left(  w_{3}+w_{5}-4g_{2}\right)  x_{-}^{i}\right)  \\
& +\frac{1}{4}x_{+(s}x_{-r)}\left(  -\frac{1}{2}r_{4}\tilde{k}^{i}+2s_{4}%
x_{+}^{i}+2w_{4}x_{-}^{i}\right)  +\frac{1}{4}x_{+(s}x_{-r)}\left(  -\frac
{1}{2}r_{6}\tilde{k}^{i}+2s_{6}x_{+}^{i}+2w_{6}x_{-}^{i}\right)  \\
& +\tilde{k}_{(s}x_{+r)}\left(  -\frac{1}{2}f_{1}\tilde{k}^{i}+2a_{3}x_{+}%
^{i}+2m_{1}x_{-}^{i}\right)  +\tilde{k}_{(s}x_{-r)}\left(  -\frac{1}{2}%
f_{2}\tilde{k}^{i}+2l_{1}x_{+}^{i}+2a_{4}x_{-}^{i}\right).
\end{align*}




This expression include all possible reductions, the psedo-differential, and the differential ones (if they exist).
In addition, we obtain a set of uniform hyperbolizations, (making zero all  coefficients with the exception of $a_i$, and  choosing them real, different among them and different to the generalized eigenvalues). To show this, we consider these particular reductions, obtain the reduced system and check the conditions between the canonical angles in  theorem \ref{Theorem_1}. Indeed, using  equation (\ref{projector_1}) in appendix \ref{appendix_a}, we conclude that equation (\ref{angles_1}) holds. The eigenprojectors are 
{\scriptsize
\[
P^{\lambda =0}\left( \tilde{k}\right) =\left[ 
\begin{array}{cc}
\tilde{k}_{(i}x_{+j)} & \tilde{k}_{(i}x_{-j)} \\ 
0 & 0%
\end{array}%
\right] \left[ 
\begin{array}{cc}
2\tilde{k}^{(s}x_{+}^{r)} & 0 \\ 
2\tilde{k}^{(s}x_{-}^{r)} & 0%
\end{array}%
\right] =\left[ 
\begin{array}{cc}
2\tilde{k}_{(i}x_{+j)}\tilde{k}^{(s}x_{+}^{r)}+2\tilde{k}_{(i}x_{-j)}\tilde{k%
}^{(s}x_{-}^{r)} & 0 \\ 
0 & 0%
\end{array}%
\right] 
\]%
\[
\left\vert P^{\lambda =0}\right\vert =\frac{1}{\cos \theta _{2}^{\lambda =0}}%
=1
\]

\begin{eqnarray*}
&&P^{\lambda = \pm 1}\left( \tilde{k}\right) = \\
&&\left[ 
\begin{array}{cc}
\frac{1}{4}\left( x_{+i}x_{+j}-x_{-i}x_{-j}\right) \left(
x_{+}^{s}x_{+}^{r}-x_{-}^{s}x_{-}^{r}\right) +\frac{1}{4}%
x_{+(i}x_{-j)}x_{+}^{(s}x_{-}^{r)} & \mp \frac{1}{2}\left(
x_{+i}x_{+j}-x_{-i}x_{-j}\right) \left(
x_{+}^{s}x_{+}^{r}-x_{-}^{s}x_{-}^{r}\right) \mp\frac{1}{2}%
x_{+(i}x_{-j)}x_{+}^{(s}x_{-}^{r)} \\ 
\mp\frac{1}{8}\left( x_{+i}x_{+j}-x_{-i}x_{-j}\right) \left(
x_{+}^{s}x_{+}^{r}-x_{-}^{s}x_{-}^{r}\right) \mp\frac{1}{8}%
x_{+(i}x_{-j)}x_{+}^{(s}x_{-}^{r)} & \frac{1}{4}\left(
x_{+i}x_{+j}-x_{-i}x_{-j}\right) \left(
x_{+}^{s}x_{+}^{r}-x_{-}^{s}x_{-}^{r}\right) +\frac{1}{4}%
x_{+(i}x_{-j)}x_{+}^{(s}x_{-}^{r)}%
\end{array}%
\right] 
\end{eqnarray*}%

\[
\left\vert P^{\lambda =\pm 1}\left( \tilde{k}\right) \right\vert =\frac{1}{\cos
\theta _{2}^{\lambda = \pm 1}}=\sqrt{\frac{1}{8}\sqrt{2}\sqrt{17}+\frac{25}{32}}%
\approx 1. 2289
\]
}

this proves that the cosines of the principal angles are uniformly lower bounded.



\subsubsection{BSSN}

The 3+1 linearized, principal part of the BSSN equations, for the fields $\left(  h_{ij},k_{ij}%
,f_{i}\right)  $ with $i,j=1,2,3$ are,
\begin{align*}
\frac{1}{N}\partial_{t}h_{ij}  &  =-2k_{ij}\\
\frac{1}{N}\partial_{t}k_{ij}  &  =-\frac{1}{2}e^{kl}\left(  \partial
_{k}\partial_{l}h_{ij}+b\partial_{i}\partial_{j}h_{kl}\right)  +\partial
_{(i}f_{j)}\\
\frac{1}{N}\partial_{t}f_{i}  &  =-2e^{kj}\partial_{k}k_{ij}+e^{kj}%
\partial_{i}k_{kj}\\
0  &  =-e^{ij}\Delta h_{ij}+\partial^{l}\partial^{j}h_{jl}\\
0  &  =\partial^{l}k_{li}-e^{kl}\partial_{i}k_{kl}\\
0  &  =\partial_{i}f_{j}-\partial_{i}\partial^{l}h_{lj}+\frac{1}{2}%
e^{kl}\partial_{i}\partial_{j}h_{kl}%
\end{align*}
Where we are using the notation of the ADM case. Therefore the principal
symbol is:
\begin{align*}
\mathfrak{N}_{~\eta}^{Ab}n\left(  \lambda\right)  _{b}\delta\phi^{\alpha}  &
=\left(  \lambda E_{\text{ }\alpha}^{A~~}+B_{~\alpha}^{A~~}\right)  \left(
\begin{array}
[c]{c}%
\hat{l}_{sr}\\
\hat{k}_{sr}\\
\hat{f}_{s}%
\end{array}
\right) \\
&  =\left(
\begin{array}
[c]{ccc}%
\lambda\delta_{(i}^{s}\delta_{j)}^{r} & 2\delta_{(i}^{s}\delta_{j)}^{r} & 0\\
\frac{1}{2}\left(  \delta_{(i}^{s}\delta_{j)}^{r}+b\tilde{\beta}_{i}%
\tilde{\beta}_{j}e^{sr}\right)  & \lambda\delta_{(i}^{s}\delta_{j)}^{r} &
-\tilde{\beta}_{(i}\delta_{j)}^{s}\\
0 & \left(  2\tilde{\beta}^{j}\delta_{(i}^{s}\delta_{j)}^{r}-e^{sr}%
\tilde{\beta}_{i}\right)  & \lambda\delta_{i}^{s}\\
-\tilde{\beta}_{i}\tilde{\beta}^{(s}\delta_{j}^{r)}+\frac{1}{2}e^{sr}%
\tilde{\beta}_{i}\tilde{\beta}_{j} & 0 & \tilde{\beta}_{i}\delta_{j}^{s}\\
-q^{sr} & 0 & 0\\
0 & \left(  \tilde{\beta}^{l}\delta_{(i}^{s}\delta_{l)}^{r}-e^{sr}\tilde
{\beta}_{i}\right)  & 0
\end{array}
\right)  \left(
\begin{array}
[c]{c}%
\hat{l}_{sr}\\
\hat{k}_{sr}\\
\hat{f}_{s}%
\end{array}
\right)  =0
\end{align*}
for the variables $
\hat{l}_{sr}, \hat{k}_{sr}, \hat{f}_{s}$ . For the details about the pseudo-differential structure see
\cite{nagy2004strongly}.

Analogously to the ADM case,
\begin{align*}
&  \sqrt{\det\left(  G^{\gamma\alpha}\overline{\left(  \lambda E_{\text{
}\alpha}^{A}+B_{~\alpha}^{A}\right)  }G_{AB}\left(  \left(  \lambda
E_{\text{~}\rho}^{B}+B_{~\rho}^{B}\right)  \right)  \right)  }\\
&  = \left\vert \lambda\right\vert^{2} \left\vert \left(  \lambda
^{2}-1\right)  \right\vert^{2} \left\vert \left(  b-\lambda^{2}\right)
\right\vert \sqrt{p_{2}\left(  \lambda,\bar{\lambda}\right)  }=0%
\end{align*}

Again,  $G^{\gamma\alpha}$ and  $G_{AB}$ are  positive definite Hermitian forms,
and $p_{2}\left(  \lambda,\bar{\lambda}\right)  $ is a positive polynomial in $\lambda$, $\bar{\lambda}.$
Thus, the modes propagation-velocities, or equivalently the
generalized eigenvalues are the same to the ADM case $\lambda=0,$ $\lambda=\pm1$ and $\lambda=\pm\sqrt{b}.$

We present the Kronecker structure of this system, but we study two cases separated, the densitized $b\neq0$ and non densitized $b=0$ lapse.

\subsection{Case $b\neq0$}

The Left kernel of $M_{~\alpha}^{A}$ are covectors of the form $\delta
X_{A}=\left(
\begin{array}
[c]{cccccc}%
\alpha^{ij}, & \gamma^{ij}, & \eta^{i} & \varepsilon^{ij}, & \rho, & \tau^{i}%
\end{array}
\right)  $ with the index $i,j=1,2,3$. We use the same definition than before for $x_{\pm}$ with $x_{\pm}.x_{\pm}=\delta_{\pm}\,\ $ and $\tilde{k
}.x_{\pm}=0,$. In addition we define the projector\ \ $q^{ij}:=e^{ij}%
-\tilde{k}^{i}\tilde{k}^{j}.$

One base of the generalized left kernel is:

For the generalized eigenvalue $\lambda=0$ we have two eigenvectors%
\[
\upsilon_{1\pm A}=\left(
\begin{array}
[c]{cccccc}
0, & 0, & x_{\pm}^{i}, & 0, & 0, & -2x_{\pm}^{i},
\end{array}
\right)
\]
For the eigenvalues $\lambda=\pm 1$ we have two eigenvectors for each case:
\begin{align*}
\upsilon_{2\pm A}  &  =\left(
\begin{array}
[c]{cccccc}%
\mp\frac{1}{2}\left(  x_{+}^{i}x_{+}^{j}-x_{-}^{i}x_{-}^{j}\right) , 
& \left(x_{+}^{i}x_{+}^{j}-x_{-}^{i}x_{-}^{j}\right) , & 0, & 0, & 0, & 0
\end{array}
\right) \\
\upsilon_{3\pm A}  &  =\left(
\begin{array}
[c]{cccccc}%
\mp\frac{1}{2}x_{+}^{(i}x_{-}^{j)}, & x_{+}^{(i}x_{-}^{j)}, & 0, & 0, & 0, & 0
\end{array}
\right)
\end{align*}


For  $\lambda=\pm\sqrt{b}$ 
\[
\upsilon_{4\pm A}=\left(
\begin{array}
[c]{cccccc}%
\mp\frac{1}{2}\sqrt{b}\tilde{k}^{i}\tilde{k}^{j} & \tilde{k}%
^{i}\tilde{k}^{j} & 0 & \tilde{k}^{i}\tilde{k}^{j} & \frac{1}{2}(b+1) & 0
\end{array}
\right)
\]



A base of the left kernel, for all $\lambda$, is composed by a set of vectors, of order $1$ in $\lambda$ and another set of order $0$ in $\lambda$.
The first ones define  $7 \times L_{1}^{T}$ blocks. They are: 

\begin{align*}
\left(  \theta_{1}\right)  _{1A}-\lambda\left(  \theta_{0}\right)  _{1A}  &
=\left(
\begin{array}
[c]{cccccc}%
q^{ij} & 0 & 0 & 0 & \lambda & 2\tilde{k}^{i}%
\end{array}
\right) \\
\left(  \theta_{1}\right)  _{2A}-\lambda\left(  \theta_{0}\right)  _{2A}  &
=\left(
\begin{array}
[c]{cccccc}%
0 & q^{ij} & 0 & 0 & \frac{1}{2} & \lambda\tilde{k}^{i}%
\end{array}
\right) \\
\left(  \theta_{1}\right)  _{3A}-\lambda\left(  \theta_{0}\right)  _{3A}  &
=\left(
\begin{array}
[c]{cccccc}%
-e^{ij} & 0 & -2\tilde{\omega}^{i} & 2\lambda\tilde{k}^{i}\tilde{k
}^{j} & 0 & 0
\end{array}
\right) \\
\left(  \theta_{1}\right)  _{4\pm A}-\lambda\left(  \theta_{0}\right)  _{4\pm
A}  &  =\left(
\begin{array}
[c]{cccccc}%
\tilde{k}^{i}x_{\pm}^{j} & 0 & -x_{\pm}^{i} & \lambda\tilde{k}%
^{i}x_{\pm}^{j} & 0 & 0
\end{array}
\right) \\
\left(  \theta_{1}\right)  _{5\pm A}-\lambda\left(  \theta_{0}\right)  _{5\pm
A}  &  =\left(
\begin{array}
[c]{cccccc}%
0 & -\tilde{k}^{i}x_{\pm}^{j} & 0 & -\frac{1}{2}\tilde{k}^{i}x_{\pm
}^{j} & 0 & \lambda x_{\pm}^{i}%
\end{array}
\right)
\end{align*}

The remaining vectors defines $6 \times L_{0}^{T}$ blocks, They are: 

\begin{align*}
\left(  \theta_{0}\right)  _{1A}  &  =\left(
\begin{array}
[c]{cccccc}%
0 & 0 & 0 & x_{+}^{i}x_{+}^{j}-x_{-}^{i}x_{-}^{j} & 0 & 0
\end{array}
\right) \\
\left(  \theta_{0}\right)  _{2A}  &  =\left(
\begin{array}
[c]{cccccc}%
0 & 0 & 0 & x_{+}^{(i}x_{-}^{j)} & 0 & 0
\end{array}
\right) \\
\left(  \theta_{0}\right)  _{3A}  &  =\left(
\begin{array}
[c]{cccccc}%
0 & 0 & 0 & x_{+}^{[i}x_{-}^{j]} & 0 & 0
\end{array}
\right) \\
\left(  \theta_{0}\right)  _{4\pm A}  &  =\left(
\begin{array}
[c]{cccccc}%
0 & 0 & 0 & \tilde{k}^{j}x_{\pm}^{i} & 0 & 0
\end{array}
\right) \\
\left(  \theta_{0}\right)  _{5A}  &  =\left(
\begin{array}
[c]{cccccc}%
0 & 0 & 0 & \frac{1}{2}q^{ij} & 0 & 0
\end{array}
\right)
\end{align*}

Therefore the Kronecker decomposition is $2 \times J_{1}(0)$, $2 \times J_{1}(1)$, 
$2 \times J_{1}(-1)$, $1 \times J_{1}(\sqrt{b})$, $1 \times J_{1}(-\sqrt{b})$, $7 \times L_{1}^{T}$ and $6 \times L_{0}^{T}$ (six null rows).

The explicitly form of the reductions that appear in Nagy et.al. work. We notices that in these reductions, the case $b=1$ leads to a weakly
hyperbolic system for $c\neq2$. The problem is the particular reduction that they are considering. Indeed, as can be easily checked,   
for the left-kernel base presented before, the Kronecker decomposition does not change in the $b=1$  case.
That is, there are other reductions that lead to a strongly hyperbolic
systems, for example choosing  $c=2$. It means that, there are not a physic problem for that case.

\subsection{Non densitized case $b=0$}

In this case, the generalized eigenvalues $\pm\sqrt{b}$ go to $0$, and the
corresponding eigenspace became degenerated, since $\upsilon_{6A}$ collapses
to $\upsilon_{7A}$ and non other eigenvector appears. We will show that
the Kronecker decomposition of $M$ have a Jordan block of order $2$ associated
to this degeneration. We follow theorem 7 in \cite{abalos2017necessary} and
show that the operator $L$ have right kernel.

To construct $L$ we calculate a base of the right kernel of $M$ for $\lambda=0$, it is
\[
\chi_{q}^{\alpha}=\left(
\begin{array}
[c]{ccc}%
2\tilde{k}_{(s}x_{+r)} & 2\tilde{k}_{(s}x_{-r)} & \tilde{k}%
_{s}\tilde{k}_{r}\\
0 & 0 & 0\\
x_{+s} & x_{-s} & \frac{1}{2}\tilde{k}_{s}%
\end{array}
\right)
\]
each column is an vector of the base, indices by $q=1,2,3.$
A base of the left kernel is%

{\scriptsize
\[
\Xi_{A}^{u}=\left(
\begin{array}
[c]{cccccc}%
0 & 0 & x_{+}^{i} & 0 & 0 & -2x_{+}^{i}\\
0 & 0 & x_{-}^{i} & 0 & 0 & -2x_{-}^{i}\\
0 & \tilde{k}^{i}\tilde{k}^{j} & 0 & \tilde{k}^{i}\tilde{k
}^{j} & +\frac{1}{2} & 0\\
\tilde{k}^{(i}x_{+}^{j)} & 0 & -x_{+}^{i} & 0 & 0 & 0\\
\tilde{k}^{(i}x_{-}^{j)} & 0 & -x_{-}^{i} & 0 & 0 & 0\\
q^{ij} & 0 & 0 & 0 & 0 & 2\tilde{k}^{i}\\
0 & q^{ij} & 0 & 0 & \frac{1}{2} & 0\\
-e^{ij} & 0 & -2\tilde{k}^{i} & 0 & 0 & 0\\
0 & -\tilde{k}^{(i}x_{+}^{j)} & 0 & -\frac{1}{2}\tilde{k}^{i}x_{+}^{j}
& 0 & 0\\
0 & -\tilde{k}^{(i}x_{-}^{j)} & 0 & -\frac{1}{2}\tilde{k}^{i}x_{-}^{j}
& 0 & 0\\
0 & 0 & 0 & \frac{q^{ij}}{2} & 0 & 0\\
0 & 0 & 0 & x_{+}^{i}\tilde{k}^{j} & 0 & 0\\
0 & 0 & 0 & x_{-}^{i}\tilde{k}^{j} & 0 & 0\\
0 & 0 & 0 & x_{+}^{[i}x_{-}^{j]} & 0 & 0\\
0 & 0 & 0 & x_{+}^{i}x_{+}^{j}-x_{-}^{i}x_{-}^{j} & 0 & 0\\
0 & 0 & 0 & x_{+}^{(i}x_{-}^{j)} & 0 & 0
\end{array}
\right)
\]}
Each row is a vector of the base, indices by $u=1,....,16.$ Then $L_{~q}^{u}$
is
{\scriptsize
\[
L_{~q}^{u}=\Xi_{A}^{u}M_{1\alpha ij}^{A~~}\chi_{q}^{\alpha}=\left(
\begin{array}
[c]{ccc}%
1 & 0 & 0\\
0 & 1 & 0\\
0 & 0 & 0\\
0 & 0 & 0\\
0 & 0 & 0\\
0 & 0 & 0\\
0 & 0 & 0\\
0 & 0 & 0\\
0 & 0 & 0\\
0 & 0 & 0\\
0 & 0 & 0\\
0 & 0 & 0\\
0 & 0 & 0\\
0 & 0 & 0\\
0 & 0 & 0\\
0 & 0 & 0
\end{array}
\right)
\]
}

Which have a trivial right kernel, therefore the Kronecker decomposition have
at least a $l-$Jordan block with $l\geq2$. In this particular case, using the dimension of the left kernel for the other generalized eigenvalues, it is possible to conclude that a $J_{2}(0)$ Jordan block appears. So this system is intrinsically ill posed, i.e.  there are not any hyperbolic
reduction as follow of Lemma  \ref{main_lema}.


\section{Discussion}


In this work we have found a necessary and sufficient condition that a first order system has to satisfy in order to have a reduction which is a hyperbolization (see Theorem \ref{Theorem_FL}). 
In the case of constant coefficient system this hyperbolization implies the reduced system has a well posed initial value formulation. 
That is, given the values of the unknowns in an appropriate hyper-surface (the initial data) a unique solution to the reduced system exists and it depends continuously on that data. 
Contrary to the classical treatment, the reduced system is not a partial differential system, but in general it is a pseudo-differential system. Nevertheless the usual theory applies in the sense that energy norms can be constructed and the corresponding estimates obtained. 
It is important to realize that once initial data is given the solution of the reduced system is unique, so if the complete system has a solution for that data, it must be that one. This in general does not need to be the case, even in the event the data chosen satisfies initially the whole system. There remains, in this more general setting, to show constraint propagation is consistent. But this problem involves looking at lower order terms of the system, the integrability conditions,  which we are avoiding  here and which in general are difficult to deal with in general. At the principal symbol level there are two approaches that can be taken regarding this problem, one is to define constraint quantities, that is, linear combinations of the fields which vanish when the constraints are satisfied and check that they also satisfy hyperbolic equations provided some consistency conditions are satisfied. This has been done partially in  \cite{reula2004strongly} for the case of algebraic reductions. The present case can be dealt with the same machinery we have used here and will be developed in a future paper.
Another approach is to extend the system, adding new variables, so that it no longer has constraints. This scheme has many advantages, in particular for obtaining systems with better numerical behavior \cite{dedner2002hyperbolic, munz2000three, brodbeck1999einstein, alic2012conformal, gundlach2005constraint}. The present machinery allows to tell when a given system would admit a hyperbolic extension, in general pseudo-differential. Again, these results will be present elsewhere.

As mentioned the resulting set of evolution equations, that a hyperbolization selects, is in general a pseudo-differential system. Thus, it is not clear whether this system has a causal propagation, that is whether there exists a maximal propagation speed. Clearly the eigenvalues are all finite, but that does not necessarily means that a solution for a compactly supported initial data would remain so. In fact if the reduced system is not analytic in $k_a$, then the solution can not have compact support. 
Indeed assume a data of compact support, $\phi^{\alpha}_0$, then its Fourier transform, $\hat{\phi}^{\alpha}_0$ is analytic. Writing the system as,
\[
\frac{d\hat{\phi}^{\alpha}}{dt} = iA^{\alpha}{}_{\beta}(k) \hat{\phi}^{\beta}.
\]
The solution would be, 
\[
\hat{\phi}^{\alpha}(k,t) = e^{iA(k)t} \hat{\phi}^{\alpha}_0(k),
\]
but if $A^{\alpha}{}_{\beta}(k)$ is not analytic, neither would be the solution for any finite $t$. Thus, for non-analytic reductions, the solution can not have compact support at any time before or after the initial slice. We believe that causality would follow for analytic reductions, a possible way to see this is using the ideas in \cite{rauch2005precise}. In any case it is then important to develop a theory spelling out necessary and sufficient conditions equivalent to the existence of analytic reductions. This would not only be important for causality, but also for extending the results of this work to variable coefficients or quasi-linear system. In that case, with the present technology, smoothness of the principal symbol is needed to obtain the energy estimates used for showing existence for such systems. 

We now turn to variable coefficient systems, a intermediate step to treat general 
quasi-linear systems. Given any point in the cotangent-bundle $(p,k)$ we can perform the Kronecker decomposition and find suitable hyperbolic reductions. But in general the reduction would also depends on the point and the bundle. Thus, when going back to space-time we would end up with an operator equation. It is not clear in what sense that equation is related to the original partial differential equation. Further ideas are needed to establish some equivalence of operators at the pseudo-differential level, and perhaps an intermediate operator which can relate our reduction to the original system. 

In the case of variable coefficient or quasi-linear systems we can still consider reductions which do not depend on $k$, $\tilde{h}_{~B}^{\gamma }$, namely a differential reduction. These  $\tilde{h}_{~B}^{\gamma}$ selects a set of evolution equations as:
\begin{equation}
\tilde{h}_{~B}^{\gamma}\mathfrak{N}_{~\alpha}^{Aa}\left( x,\phi\right)
\nabla_{a}\phi^{\alpha}=\tilde{h}_{~B}^{\gamma}J^{A}\left( x,\phi\right) . 
\label{ql_2}
\end{equation}
Assume our system satisfies all the uniformity conditions for all points and furthermore that we can find a differential reduction $\tilde{h}_{~B}^{\gamma }\left( x,\phi_{0}\right) $ among all possible, with $\phi _{0}^{\alpha }$ a background solution, then the system (\ref{ql_2}) is uniformly diagonalizable for all points $\left(p,k_{a}\right) $. Furthermore, if the symmetrizer, $\overline{W_{~\alpha }^{\delta }}H_{\delta C}I_{~\gamma }^{C}W_{~\eta }^{\gamma }$  is smooth in their variables $\left(p,\phi ^{\alpha },k_{a}\right) $, then system (\ref{ql_2}) is strongly hyperbolic, see \cite{sarbach2012continuum}.
In addition if $\tilde{h}_{~B}^{\gamma }\left(p,\phi _{0}\right) $ is a differential reduction and $\overline{W_{~\alpha }^{\delta }}H_{\delta
C}I_{~\gamma }^{C}W_{~\eta }^{\gamma }$ is independent of $k_{a}$ and smooth in their variables, the systems is symmetric hyperbolic. 

\section*{Acknowledgements}
We would like to thank Robert Geroch, Marcelo Rubio, Federico Carrasco and Miguel Megevand for the discussions and ideas exchanged throughout this work.
We acknowledge financial support from CONICET, SECYT-UNC and FONCYT Argentina.

\appendix


\section{Proof of theorem \ref{Theorem_1} \label{appendix_a}}

Consider the resolvent Kreiss condition of the matrix theorem eq. (\ref%
{kreiss_1}). We shall use it in the following equivalent form,
\begin{equation}
\frac{1}{C}\varepsilon \leq \min_{j\in \left\{ 1,..,u\right\} }\sigma _{j}%
\left[ A_{~\gamma }^{\alpha a}k_{a}-\lambda _{R}\delta _{~\gamma }^{\alpha
}-i\varepsilon \delta _{~\gamma }^{\alpha }\right]   \label{kreiss_sing_1}
\end{equation}%
where $s=\lambda _{R}+i\varepsilon $ and we have used that for any invertible 
$B\in \mathbb{C}^{u\times u}$ matrix, $\left\vert B^{-1}\right\vert =%
\frac{1}{\underset{j\in \left\{ 1,..,u\right\} }{\min }\sigma _{j}\left[ B%
\right] }$, with $\sigma _{j}\left[ B\right] $ the singular values of $B$.

We are going to prove that equation (\ref{kreiss_sing_1}) holds for all $%
\lambda _{R}\in \mathbb{R}$ and all $\varepsilon >0$ if and only if all the
eigenvalues $\tau _{i}\left( k\right) $ of $\ A_{~\gamma }^{\alpha a}k_{a}$
are real, and for all $i\in F_{\left( k\right) }=\left\{
1,...,w_{\left( k\right) }\right\} $  and all $k_{a}$ non
proportional to $n_{a}$, with $\left\vert k\right\vert =1$, the angles $\theta
_{l}^{\tau _{i}\left( k\right) }$ between $\Upsilon _{L}^{\tau _{i}\left(
k\right) }$ and $\Phi _{R}^{\tau _{i}\left( k\right) }$ hold the lower
bound condition 
\begin{equation}
 \cos \theta _{l}^{\tau _{i}\left( k\right)}	\geq \cos \vartheta  > 0
\label{bb_1}
\end{equation}

$\Longleftarrow )$ \ Consider the right hand size of equation (\ref%
{kreiss_sing_1}), with $\lambda _{R}=\tau _{i}\left( k\right) $ and $0\leq
\varepsilon <<1$. In that case, as it was explained in \cite%
{abalos2017necessary}, the vanishing (at $\varepsilon =0$) singular values
have the following $\varepsilon $ dependence: 
\begin{align}
\sigma _{u-r_{\tau _{i}\left( k\right) }+l}\left[ A_{~\gamma }^{\alpha
a}k_{a}-\tau _{i}\left( k\right) \delta _{~\gamma }^{\alpha }-i\varepsilon
\delta _{~\gamma }^{\alpha }\right] & =\varepsilon \text{ }\sigma _{l}\left[
T^{\tau _{i}\left( k\right) }\right] +O\left( \varepsilon ^{2}\right)  
\notag \\
& =\varepsilon \cos \theta _{l}^{\tau _{i}\left( k\right) }+O\left(
\varepsilon ^{2}\right)   \label{cos_1}
\end{align}%
with $l \in I_{^{\tau _{i}\left( k\right) }}:=\left\{ 1,..,r_{\tau _{i}\left(
k\right) }\right\} $ and $\left( T^{\tau _{i}\left( k\right) }\right)
_{~j}^{i}$ is given in equation (\ref{eq_Tij_1}). That is, $\sigma _{l}\left[
T^{\tau _{i}\left( k\right) }\right] =\cos \theta _{l}^{\tau _{i}\left(
k\right) }$, where $\theta _{l}^{\tau _{i}\left( k\right) }$ are the angles
between the subspaces $\Upsilon _{L}^{\tau _{i}\left( k\right) }$ and $\Phi
_{R}^{\tau _{i}\left( k\right) }$. Since by hypothesis above these cosines
are bounded by below (eq. (\ref{bb_1})), (for all $\tau _{i}\left( k\right) $%
, $l$ and $k_{a}$) , this means that the singular values of $A_{~\gamma
}^{\alpha a}k_{a}-\tau _{i}\delta _{~\gamma }^{\alpha }-i\varepsilon \delta
_{~\gamma }^{\alpha }$ are all order $O\left( \varepsilon ^{0}\right) $ or $%
O\left( \varepsilon ^{1}\right) $. So by theorem 3.3. in \cite%
{abalos2017necessary} the matrix $A_{~\gamma }^{\alpha a}k_{a}$ is
diagonalizable for all $k_{a}$ not proportional to $n_{a}$ and $\left\vert
k\right\vert =1$. Thus eq. (\ref{kreiss_sing_1}) holds with $\frac{1}{C}%
=\cos \vartheta $ and $0\leq \varepsilon <<1$. We now extend the proof for
all $\varepsilon $. To do that, we shall take the limit when $\varepsilon $
goes to zero in eq. (\ref{kreiss_1}), in two different form, resulting in an
upper bound of the eigen-projectors of $A_{~\gamma }^{\alpha a}k_{a}$ (as Strang showed \cite{strang1967strong}) that
then we shall use to conclude the first part of the theorem. We first take, 
\begin{align}
& \underset{\varepsilon \rightarrow 0}{\lim }\left\vert \left( A_{~\gamma
}^{\alpha a}k_{a}-\tau _{i}\left( k\right) \delta _{~\gamma }^{\alpha
}-i\varepsilon \delta _{~\gamma }^{\alpha }\right) ^{-1}\right\vert
\varepsilon   \label{A_a1} \\
& =\underset{\varepsilon \rightarrow 0}{\lim }\frac{\varepsilon }{\underset{%
j\in \left\{ 1,..,u\right\} }{\min }\sigma _{j}\left[ A_{~\gamma }^{\alpha
a}k_{a}-\tau _{i}\left( k\right) \delta _{~\gamma }^{\alpha }-i\varepsilon
\delta _{~\gamma }^{\alpha }\right] }  \notag \\
& =\frac{1}{\underset{l\in I_{^{\tau _{i}\left( k\right) }}}{\min }\cos
\theta _{l}^{\tau _{i}\left( k\right) }}  \label{A_a3}.
\end{align}

On the other hand, since we know that $A_{~\gamma }^{\alpha a}k_{a}$ is
diagonalizable, it can be written in term of their eigen-projectors $\left(
P^{j}\left( k\right) \right) _{~\gamma }^{\alpha }$ 
\begin{equation*}
A_{~\gamma }^{\alpha a}k_{a}=\underset{i\in F_{\left( k\right) }}{\sum }\tau
_{j}\left( k\right) \left( P^{j}\left( k\right) \right) _{~\gamma }^{\alpha }
\end{equation*}%
where $\left( P^{i}\left( k\right) \right) _{~\gamma }^{\alpha }\left(
P^{j}\left( k\right) \right) _{~\eta }^{\gamma }=\underset{j\in F_{\left(
k\right) }}{\sum }\delta _{~j}^{i}\left( P^{j}\left( k\right) \right)
_{~\eta }^{\alpha }$ and $\underset{j\in F_{\left( k\right) }}{\sum }\left(
P^{j}\left( k\right) \right) _{~\gamma }^{\alpha }=\delta _{~\gamma
}^{\alpha }$. Since 
\begin{align*}
& \underset{\varepsilon \rightarrow 0}{\lim }\left\vert \left( A_{~\gamma
}^{\alpha a}k_{a}-\tau _{i}\left( k\right) \delta _{~\gamma }^{\alpha
}-i\varepsilon \delta _{~\gamma }^{\alpha }\right) ^{-1}\right\vert
\varepsilon \\ & =\underset{\varepsilon \rightarrow 0}{\lim }\left\vert 
\underset{j\in F_{\left( k\right) }}{\sum }\frac{\varepsilon }{\left( \tau
_{j}\left( k\right) -\tau _{i}\left( k\right) -i\varepsilon \right) }\left(
P^{j}\left( k\right) \right) _{~\gamma }^{\alpha }\right\vert  \\
& =\left\vert \left( P^{j}\left( k\right) \right) _{~\gamma }^{\alpha
}\right\vert 
\end{align*}%
We conclude, 
\begin{equation} 
\left\vert \left( P^{j}\left( k\right) \right) _{~\gamma }^{\alpha
}\right\vert =\frac{1}{\underset{l\in I_{^{\tau _{i}\left( k\right) }}}{\min 
}\cos \theta _{l}^{\tau _{i}\left( k\right) }} \label{projector_1}
\end{equation}%
and so, by eq. (\ref{bb_1}), 
\begin{equation*}
\left\vert \left( P^{j}\left( k\right) \right) _{~\gamma }^{\alpha
}\right\vert \leq \frac{1}{\cos \vartheta }=C
\end{equation*}

Thus, for any $\varepsilon >0$ 
\begin{align*}
\left\vert \left( A_{~\gamma }^{\alpha a}k_{a}-\lambda _{R}\delta _{~\gamma
}^{\alpha }-i\varepsilon \delta _{~\gamma }^{\alpha }\right)
^{-1}\right\vert \varepsilon & =\left\vert \underset{j\in F_{\left( k\right)
}}{\sum }\frac{\varepsilon }{\left( \tau _{j}\left( k\right) -\lambda
_{R}-i\varepsilon \right) }\left( P^{j}\left( k\right) \right) _{~\gamma
}^{\alpha }\right\vert  \\
& \leq \underset{j\in F_{\left( k\right) }}{\sum }\left\vert \frac{%
\varepsilon }{\left( \tau _{j}\left( k\right) -\lambda _{R}-i\varepsilon
\right) }\right\vert \left\vert \left( P^{j}\left( k\right) \right)
_{~\gamma }^{\alpha }\right\vert  \\
& \leq \underset{j\in F_{\left( k\right) }}{\sum }\left\vert \left(
P^{j}\left( k\right) \right) _{~\gamma }^{\alpha }\right\vert  \\
& \leq uC
\end{align*}%
where in third line we have used that $\frac{\left\vert \varepsilon
\right\vert }{\left\vert \left( \tau _{j}\left( k\right) -\lambda
_{R}-i\varepsilon \right) \right\vert }\leq 1$.

$\Longrightarrow )$ If Kreiss resolvent matrix condition holds we know by
definition \ref{Definition:2} that $A_{~\gamma }^{\alpha a}k_{a}$ is
diagonalizable with real eigenvalues for all $k_{a}$. Since now eq. (\ref%
{A_a1}) is upper bounded by $C$, we conclude, using eq. (\ref{A_a3}), that 
\begin{equation*}
\frac{1}{\underset{l\in I_{^{\tau _{i}\left( k\right) }}}{\min }\cos \theta
_{l}^{\tau _{i}\left( k\right) }}\leq C
\end{equation*}%
for all $i\in F_{\left( k\right) }$ and all $k_{a}$. This concludes the
proof of theorem \ref{Theorem_1}.


\section{General Kronecker decomposition.\label{appendix_II}}


For presentation in what follows we shall call these matrices, 
\begin{equation*}
\begin{array}{ccc}
E_{~\eta }^{A}:=\left( -\mathfrak{N}_{~\eta }^{Ab}n_{b}\right)  &  & 
B_{~\alpha }^{A}:=\left( \mathfrak{N}_{~\eta }^{Ab}k_{b}\right) 
\end{array}%
\end{equation*}

The full Kronecker decomposition of any pair of matrices $\left( E_{~\eta
}^{A},B_{~\alpha}^{A}\right) $ in the pencil form $\lambda
E_{~\eta}^{A}+B_{~\alpha}^{A}$ is given by, 
\begin{equation*}
\lambda E_{~\eta}^{A}+B_{~\alpha}^{A}=Y_{~B}^{A}K_{~\alpha}^{B}\left(
\lambda\right) W_{~\eta}^{\alpha}
\end{equation*}
where $Y_{~B}^{A}$ and $W_{~\eta}^{\alpha}$ are invertible independent of $%
\lambda$, and $K_{~\alpha}^{B}\left( \lambda\right) $ is a block matrix.

As we mention in section \ref{subsec_Kronecker_decomp}, when $E_{~\eta}^{A}$
have no right kernel, the blocks of $K_{~\alpha}^{B}\left( \lambda\right) $
are $J_{m}\left( \lambda _{i}\right) $-Jordan Blocks, $L_{m}^{T}$-blocks and
vanished row. In the general case, when $E_{~\eta}^{A}$ have right kernel,
new blocks appear:

The zero blocks%
\begin{equation*}
O_{m\times l}=\left( 
\begin{array}{ccc}
0 & 0 & 0 \\ 
0 & ... & 0 \\ 
0 & 0 & 0%
\end{array}
\right) \in \mathbb{C} ^{m\times l}, 
\end{equation*}
the $L_{m}$ blocks 
\begin{equation*}
L_{m}=\left( \overset{m+1}{\overbrace{%
\begin{array}{ccccc}
\lambda & 1 & 0 & 0 & 0 \\ 
0 & \lambda & 1 & 0 & 0 \\ 
0 & 0 & ... & ... & 0 \\ 
0 & 0 & 0 & \lambda & 1%
\end{array}
}}\right) \in \mathbb{C}^{m\times m+1} 
\end{equation*}

and the $N_{m}$ blocks 
\begin{equation*}
N_{m}=\left( 
\begin{array}{ccccc}
1 & \lambda & 0 & 0 & 0 \\ 
0 & 1 & \lambda & 0 & 0 \\ 
0 & 0 & ... & ... & 0 \\ 
0 & 0 & 0 & 1 & \lambda \\ 
0 & 0 & 0 & 0 & 1%
\end{array}
\right) \in \mathbb{C} ^{m\times m} 
\end{equation*}

Notice that the right kernel of $E_{~\eta}^{A}$ is given by: trivial kernel
for the zero blocks,

$\left( 
\begin{array}{c}
0 \\ 
0 \\ 
... \\ 
1%
\end{array}
\right) \in \mathbb{C} ^{m+1\times1}$ for $L_{m}$ and $\left( 
\begin{array}{c}
1 \\ 
0 \\ 
... \\ 
0%
\end{array}
\right) \in \mathbb{C}^{m\times1}$ for $N_{m}.$


\section{Kronecker decomposition of hyperbolic systems. \label{appendix}}

In this appendix, we shall show how to construct the bases in which the
system reduce to Kronecker blocks.

Consider $l\left( \lambda \right) _{a}\in S_{n_{a}}^{\mathbb{C}}$ for some $%
n_{a}.$ We shall use the left kernel as the principal ingredient for
building the bases. It is given by equation, 
\begin{equation}
X_{A}\mathfrak{N}_{~\eta }^{Ab}l\left( \lambda \right) _{b}=X_{A}\left[
\lambda \left( -\mathfrak{N}_{~\eta }^{Ab}n_{b}\right) +\left( \mathfrak{N}%
_{~\eta }^{Ab}k_{b}\right) \right] =0  \label{left_ker}
\end{equation}%
Since we are dealing with hyperbolic systems for which we already have shown
that the Jordan blocks are one dimensional we shall restrict consideration
to $-\mathfrak{N}_{~\eta }^{Ab}n_{b}$ of maximal range and $\mathfrak{N}%
_{~\eta }^{Ab}l\left( \lambda \right) _{b}$ with only 1-dimensional Jordan
Blocks. 

As appendix \ref{appendix_II}, in what follows we shall call these matrices, 
\begin{equation*}
\begin{array}{ccc}
E_{~\eta }^{A}:=\left( -\mathfrak{N}_{~\eta }^{Ab}n_{b}\right)  &  & 
B_{~\alpha }^{A}:=\left( \mathfrak{N}_{~\eta }^{Ab}k_{b}\right) 
\end{array}%
\end{equation*}

First consider the case $e=u$ where $e$ $=\dim "A"$ and $\ u=\dim "\alpha "$
for some $\left( x,\phi ,n,k\right) $ then $\mathfrak{N}_{~\eta
}^{Ab}l\left( \lambda \right) _{b}$ has a trivial Jordan decomposition,
therefore 
\begin{align}
\mathfrak{N}_{~\eta }^{Ab}l\left( \lambda \right) _{b}&
=\sum_{i=1}^{u}\upsilon ^{iA}\left( \upsilon _{iA}E_{~\alpha }^{A}\right)
\left( \lambda -\lambda _{i}\right)   \notag \\
& =\upsilon ^{1A}\left( \upsilon _{1A}E_{~\alpha }^{A}\right) \left( \lambda
-\lambda _{1}\right) +\upsilon _{2}^{A}\left( \upsilon _{2A}Y_{~\alpha
}^{A}\right) \left( \lambda -\lambda _{2}\right) +...\\  &+\upsilon ^{uA}\left(
\upsilon _{uA}E_{~\alpha }^{A}\right) \left( \lambda -\lambda _{u}\right)  
\notag \\
& =\left( 
\begin{array}{cccc}
\upsilon ^{1A} & \upsilon ^{2A} & ... & \upsilon ^{uA}%
\end{array}%
\right) \left( 
\begin{array}{cccc}
\lambda -\lambda _{1} & 0 & 0 & 0 \\ 
0 & \lambda -\lambda _{2} & 0 & 0 \\ 
0 & 0 & ... & 0 \\ 
0 & 0 & 0 & \lambda -\lambda _{u}%
\end{array}%
\right) \left( 
\begin{array}{c}
\upsilon _{1A}E_{~\alpha }^{A} \\ 
\upsilon _{2A}E_{~\alpha }^{A} \\ 
... \\ 
\upsilon _{uA}E_{~\alpha }^{A}%
\end{array}%
\right)   \label{Bases_matrix}
\end{align}

Here $\{\upsilon _{iA}$ with $i=1,...,u\}$ are the left eigen-covectors
associated to the eigenvalues $\lambda _{i}$ or the left kernel of $%
\mathfrak{N}_{~\eta }^{Ab}l\left( \lambda _{i}\right) _{b}$ i.e. 
\begin{equation}
\upsilon _{iA}\mathfrak{N}_{~\eta }^{Ab}l\left( \lambda _{i}\right) _{b}=0,
\label{eigenvectors_1}
\end{equation}%
and $\{\upsilon ^{iA}\}$ are the co-bases, 
\begin{equation*}
\upsilon ^{jA}\upsilon _{iA}=\delta _{~i}^{j}
\end{equation*}%
In eq. (\ref{Bases_matrix}) the vectors $\{\upsilon ^{jA}\}$ are in column
form and the co-vectors $\upsilon _{1A}E_{~\alpha }^{A}$ in row form. Notice
that 
\begin{align*}
\sum_{i=1}^{u}\upsilon ^{iA}\left( \upsilon _{iA}E_{~\alpha }^{A}\right) &
=\left( -\mathfrak{N}_{~\eta }^{Ab}n_{b}\right) , \\
\sum_{i=1}^{u}\upsilon ^{iA}\left( \upsilon _{iA}E_{~\alpha }^{A}\right)
\lambda _{i}& =\left( \mathfrak{N}_{~\eta }^{Ab}k_{b}\right) ,
\end{align*}%
and it is possible to find the eigenvalues $\lambda _{i}$, solving the
polynomial equation $\det \left( \mathfrak{N}_{~\eta }^{Ab}l\left( \lambda
\right) _{b}\right) =0$ for $\lambda .$

In the case in which $e>u$ for some $\left( x,\phi ,n,k\right) $, the
decomposition has, besides the Jordan blocks, additional blocks usually
denoted as $L_{m}^{T}$-blocks.
The maximal range condition on $\mathfrak{N}%
_{~\eta }^{Ab}n_{b}$ prevents other Kronecker blocks from appearing. 

As before we use the left kernel to compute the decomposition. 
For arbitrary $\lambda $, and fixed $k_{a}$, there is a left kernel subspace
which is of fixed dimension $(e-u)$. This subspace, which we call $\Delta
(\lambda )$, depends in a polynomial way with respect to $\lambda $ in fact
it is generated by a set of linearly independent vectors $\{\chi
_{A}^{i}\left( \lambda \right) \}$ with $i=1,...,e-u$ such that $\chi
_{A}^{i}\left( \lambda \right) \mathfrak{N}_{~\eta }^{Ab}l\left( \lambda
\right) _{b}=0$ $\ \forall \lambda \in \mathbb{C}$. The coefficient of the $%
\chi _{A}^{i}\left( \lambda \right) $ are polynomial in $\lambda $. Among
all the possible bases of $\Delta (\lambda )$ we take the non-zero $\chi
_{A}^{1}\left( \lambda \right) $ such that it has the least degree, that we
call $m_{1},$ in $\lambda $. It is 
\begin{equation}
\chi _{A}^{1}\left( \lambda \right) =\left( \left( \theta _{m_{1}}\right)
_{A}-\lambda \left( \theta _{m_{1}-1}\right) _{A}+\lambda ^{2}\left( \theta
_{m_{1}-2}\right) _{A}-...-\lambda ^{m_{1}}\theta _{0A}\right)   \label{XA}
\end{equation}

With $(\theta_{s})_{A}$ independent of $\lambda$. We continue chosen $%
\chi_{A}^{2}\left( \lambda\right) $ independent to $\chi_{A}^{1}\left(
\lambda\right) $ and of lest degree too, such that 
\begin{equation*}
m_{1}\leq m_{2}
\end{equation*}

We perform this process until there are no linear-independent vectors in $%
\Delta (\lambda )$. Thus we finish with a set of scalars that are called
minimal indices for the rows of $\mathfrak{N}_{~\eta }^{Ab}l\left( \lambda
\right) _{b}$ which are 
\begin{equation*}
m_{1}\leq m_{2}\leq ...\leq m_{e-u}.
\end{equation*}%
They define the $L_{m_{i}}^{T}$ blocks of $\mathfrak{N}_{~\eta }^{Ab}l\left(
\lambda \right) _{b}.$

In particular for each $\chi _{A}^{i}$  as in eq. (\ref{XA}) holding
eq. (\ref{left_ker}) follows 

\begin{eqnarray}
\left( \theta _{m_{1}}\right) _{A}B_{~\alpha }^{A} &=&0\text{ \ \ \ \ \ \ \
\ \ \ \ \ \ \ \ \ }\left( \theta _{m_{1}-1}\right) _{A}B_{~\alpha
}^{A}=\left( \theta _{m_{1}}\right) _{A}E_{~\eta }^{A}\text{ \ \ \ }... 
\notag \\
\left( \theta _{0}\right) _{A}B_{~\alpha }^{A} &=&\left( \theta _{1}\right)
_{A}E_{~\eta }^{A}\text{ \ \ \ \ \ \ \ \ \ \ \ }\left( \theta _{0}\right)
_{A}E_{~\eta }^{A}=0  \label{eq_tita}
\end{eqnarray}

The example in eq. (\ref{example_1}) has $m_{1}=0,$ $m_{2}=0$ because the
two null rows and $m_{3}=1,$ $m_{4}=2\ $\ since it has an $L_{1}^{T}$ and $%
L_{2}^{T}$ blocks

In addition, when $\lambda = \lambda _{i}$ \footnote{We recall that the $\lambda _{i}$ can be calculated
solving equation (\ref{p_1}).} the generalized eigen-values, the left kernel increases with a set of new independent co-vectors called $\{v_{iA}\}$. That  elements are only defined up to
members of $\{\chi
_{A}^{i}\left( \lambda_i \right) \}$.

Collecting all the $\{v_{iA}\}$ associated to the different generalized eigenvalues, and the co-vectors $\{\chi
_{A}^{i}\left( \lambda \right) \}$, it is possible to compute the Kronecker decomposition as in the following example.

\begin{align*}
& \mathfrak{N}_{~\eta }^{Ab}l\left( \lambda \right) _{b} = \\
&   {\footnotesize \left( 
\begin{array}{ccccccccc}
\upsilon ^{1A} & \upsilon ^{2A} & \upsilon ^{3A} & \theta ^{3A} & \theta
^{2A} & \theta ^{1A} & \theta ^{0A} & \tilde{\theta}^{0A} & \hat{\theta}^{0A}%
\end{array}%
\right) \left( 
\begin{array}{cccccc}
\lambda -\lambda _{1} & 0 & 0 & 0 & 0 & 0 \\ 
0 & \lambda -\lambda _{2} & 0 & 0 & 0 & 0 \\ 
0 & 0 & \lambda -\lambda _{3} & 0 & 0 & 0 \\ 
0 & 0 & 0 & \lambda  & 0 & 0 \\ 
0 & 0 & 0 & 1 & \lambda  & 0 \\ 
0 & 0 & 0 & 0 & 1 & \lambda  \\ 
0 & 0 & 0 & 0 & 0 & 1 \\ 
0 & 0 & 0 & 0 & 0 & 0 \\ 
0 & 0 & 0 & 0 & 0 & 0%
\end{array}%
\right) \left( 
\begin{array}{c}
\upsilon _{1C}E_{~\alpha }^{C} \\ 
\upsilon _{2C}E_{~\alpha }^{C} \\ 
\upsilon _{3C}E_{~\alpha }^{C} \\ 
\theta _{3C}E_{~\alpha }^{C} \\ 
\theta _{2C}E_{~\alpha }^{C} \\ 
\theta _{1C}E_{~\alpha }^{C}%
\end{array}%
\right) }
\end{align*}%

\begin{align*}
& =\sum_{i=1}^{3}\upsilon ^{iA}\left( \upsilon _{iA}E_{~\alpha }^{A}\right)
\left( \lambda -\lambda _{i}\right) +\lambda \sum_{i=1}^{m_{3}=3}\theta
^{iA}\left( \theta _{iC}E_{~\alpha }^{C}\right) +\sum_{i=0}^{m_{3}=3}\theta
^{iA}\left( \left( \theta _{i+1}\right) _{C}E_{~\alpha }^{C}\right)  \\
& =\lambda \lbrack \upsilon ^{1A}\left( \upsilon _{1C}E_{~\alpha
}^{C}\right) +\upsilon ^{2A}\left( \upsilon _{2C}E_{~\alpha }^{C}\right)
+\upsilon ^{3A}\left( \upsilon _{3C}E_{~\alpha }^{C}\right)  \\
& +\theta ^{3A}\left( \theta _{3C}E_{~\alpha }^{C}\right) +\theta
^{2A}\left( \theta _{2C}E_{~\alpha }^{C}\right) +\theta ^{1A}\left( \theta
_{1C}E_{~\alpha }^{C}\right) ] \\
& -\lambda _{1}\upsilon ^{1A}\left( \upsilon _{1C}E_{~\alpha }^{C}\right)
-\lambda _{2}\upsilon ^{2A}\left( \upsilon _{2C}E_{~\alpha }^{C}\right)
-\lambda _{3}\upsilon ^{3A}\left( \upsilon _{3C}E_{~\alpha }^{C}\right)  \\
& +\theta ^{2A}\left( \theta _{3C}E_{~\alpha }^{C}\right) +\theta
^{1A}\left( \theta _{2C}E_{~\alpha }^{C}\right) +\theta ^{0A}\left( \theta
_{1C}E_{~\alpha }^{C}\right) 
\end{align*}

Here $m_{1}=0$, $m_{2}=0$, $m_{3}=3$ and $v_{iA}$ are the left
eigen-covectors associated to the generalized eigenvalues $\lambda_{i}$, $%
\hat{\theta}_{0A}$ define the $L_{m_{1}}^{T}$, $\tilde{\theta}_{0A}$\ the $%
L_{m_{2}}^{T}$ and $\theta_{iA}$ with $i=0,...,3$ define the $L_{m_{3}}^{T}$
block. Here the $\lambda_{i}$ can be degenerated, i.e. they can  take the same values. 

The vector with the raised indices are the co-bases,%
\begin{align*}
v^{iA}v_{jA} & =\delta_{~j}^{i}\text{ with }i,j=1,2,3 \\
\theta^{iA}\theta_{jA} & =\delta_{~j}^{i}\text{ with }i,j=1,2,3 \\
\hat{\theta}^{0A}\hat{\theta}_{0A} & =1 \\
\tilde{\theta}^{0A}\tilde{\theta}_{0A} & =1
\end{align*}
and any other contraction vanish.


\section{Reductions in $L^{T}_{1}$ case.\label{appendix_V}}


Consider a set of $L_{1}^{T}$-blocks in the Kronecker decomposition of a
principal symbol, we can reduce all of them together and find more general 
reductions with diagonalizable reduced principal symbols, for instance given the following structure,

\begin{equation*}
\left( 
\begin{array}{cccc}
L_{1}^{T} & 0 & 0 & 0 \\ 
0 & L_{1}^{T} & 0 & 0 \\ 
0 & 0 & ... & 0 \\ 
0 & 0 & 0 & L_{1}^{T}%
\end{array}
\right) =\lambda\left( 
\begin{array}{cccc}
1 & 0 & 0 & 0 \\ 
0 & 0 & 0 & 0 \\ 
0 & 1 & 0 & 0 \\ 
0 & 0 & 0 & 0 \\ 
... & ... & ... & ... \\ 
0 & 0 & 0 & 0 \\ 
0 & 0 & 0 & 1 \\ 
0 & 0 & 0 & 0%
\end{array}
\right) +\left( 
\begin{array}{cccc}
0 & 0 & 0 & 0 \\ 
1 & 0 & 0 & 0 \\ 
0 & 0 & 0 & 0 \\ 
0 & 1 & 0 & 0 \\ 
0 & 0 & 0 & 0 \\ 
... & ... & ... & ... \\ 
0 & 0 & 0 & 0 \\ 
0 & 0 & 0 & 1%
\end{array}
\right) 
\end{equation*}

For these cases, more general $H$ are possible, they are given by, 
\begin{equation*}
\left( H_{L_{1}^{T}}\right) _{\delta A}=\left( 
\begin{array}{cccccccc}
b_{1} & 0 & \bar{b}_{2} & 0 & ... & 0 & \bar{b}_{s} & 0 \\ 
b_{2} & 0 & b_{s+1} & 0 & ... & 0 & \bar{b}_{2s-1} & 0 \\ 
... & 0 & ... & 0 & ... & 0 & ... & 0 \\ 
b_{s} & 0 & b_{2s-1} & 0 & ... & 0 & b_{\frac{s\left( s+1\right) }{2}} & 0%
\end{array}
\right) +\left( 
\begin{array}{cccccccc}
0 & c_{1} & 0 & \bar{c}_{2} & 0 & ... & 0 & \bar{c}_{s} \\ 
0 & c_{2} & 0 & c_{s+1} & 0 & ... & 0 & \bar{c}_{2s-1} \\ 
0 & ... & 0 & ... & 0 & ... & 0 & ... \\ 
0 & c_{s} & 0 & c_{2s-1} & 0 & ... & 0 & c_{\frac{s\left( s+1\right) }{2}}%
\end{array}
\right) 
\end{equation*}
Indeed, 
\begin{equation*}
\left( H_{L_{1}^{T}}\right) _{\delta A}\left( 
\begin{array}{cccc}
L_{1}^{T} & 0 & 0 & 0 \\ 
0 & L_{1}^{T} & 0 & 0 \\ 
0 & 0 & ... & 0 \\ 
0 & 0 & 0 & L_{1}^{T}%
\end{array}
\right) =\lambda\left( g_{1}\right) _{\delta\alpha}+\left( g_{2}\right)
_{\delta\alpha}, 
\end{equation*}
with, 
\begin{equation*}
\left( g_{1}\right) _{\delta\alpha}=\left( 
\begin{array}{cccc}
b_{1} & \bar{b}_{2} & ... & \bar{b}_{s} \\ 
b_{2} & b_{s+1} & ... & \bar{b}_{2s-1} \\ 
... & ... & ... & ... \\ 
b_{s} & b_{2s-1} & ... & b_{\frac{s\left( s+1\right) }{2}}%
\end{array}
\right) \text{ \ \ \ \ \ }\left( g_{2}\right) _{\delta\alpha}=\left( 
\begin{array}{cccc}
c_{1} & \bar{c}_{2} & ... & \bar{c}_{s} \\ 
c_{2} & c_{s+1} & ... & \bar{c}_{2s-1} \\ 
... & ... & ... & ... \\ 
c_{s} & c_{2s-1} & ... & c_{\frac{s\left( s+1\right) }{2}}%
\end{array}
\right) 
\end{equation*}
since both are Hermitian we only need to assert positivity of 
$(g_{1})_{\delta\alpha}$ to conclude the
reduction gives rise to a diagonalizable reduced system.
This can be done choosing appropriately the $b$ coefficients.

\section{Orthogonality of kernel subspaces.
\label{app_IIIb}}

Here we show that the subspaces $\Phi_{R}^{\lambda_{i}\left(  k\right)}$ and $\Delta_{\mathfrak{N}}(\lambda)$ are orthogonal among each other.

As before we use the notation, $\mathfrak{N}_{~\eta}^{Ab}n\left(
\lambda\right)_{b}=\left(  \lambda E_{~\eta}^{A}+B_{~\eta}^{A}\right)  $
with $l\left(  \lambda\right)_{b}\in S_{n_{b}}$. 
We are going to
show that
\begin{equation}
\chi_{A}^{s}\left(  \lambda\right)  E_{~\eta}^{A}\left(  \delta\phi_{\lambda_{i}}\right)_{s}^{\alpha}=0 \label{eq_chi_1}
\end{equation}

As it was explained in  appendix \ref{appendix}, for each $L_{m}^{T}$ block of the
principal symbol $\mathfrak{N}_{~\eta}^{Cb} l\left(  \lambda\right)_{b}$,
there are $m+1$ co-vectors $\{\theta_{s}\}$ such that
\begin{equation}
\chi_{A}^{j}\left(  \lambda\right)  =\left(  \left(  \theta_{m}\right)
_{A}-\lambda\left(  \theta_{m-1}\right)  _{A}+\lambda^{2}\left(  \theta
_{m-2}\right)  _{A}-...-\lambda^{m}\theta_{0A}\right)  \label{eq_X}%
\end{equation}
is in the left kernel for all $\lambda\in\mathbb{R}$ and the co-vectors $\{\theta_{s}\}$ fulfill the equations (\ref{eq_tita}). 
This set of co-vectors expanded the subspace $\Delta(\lambda)$. On the other hand when $\lambda=\lambda
_{i}\left(  k\right)  $ the right kernel elements satisfy,
\begin{equation}
\left(  \lambda_{i}\left(  k\right)  E_{~\eta}^{A}+B_{~\eta}^{A}\right)
\left(  \delta\phi_{\lambda_{i}}\right)  _{s}^{\eta}=0 \label{eq_phi}%
\end{equation}

In order to show eq. (\ref{eq_chi_1}), it is enough to proof that $\left(
\theta_{s}\right)  _{A}E_{~\eta}^{A}\left(  \delta\phi_{\lambda_{i}}\right)
_{t}^{\eta}=0$ for any $s$ in any $L^{T}$-block. Using equation (\ref{eq_tita}%
), (\ref{eq_phi}) several times in sequence,
\begin{align*}
\left(  \theta_{s}\right)  _{A}E_{~\eta}^{A}\left(  \delta\phi_{\lambda_{i}%
}\right)  _{t}^{\eta}  &  =\left(  \theta_{s-1}\right)  _{A}B_{~\eta}%
^{A}\left(  \delta\phi_{\lambda_{i}}\right)  _{t}^{\eta}\\
&  =-\left(  \theta_{s-1}\right)  _{A}\lambda_{i}\left(  k\right)
E_{~\eta}^{A}\left(  \delta\phi_{\lambda_{i}}\right)  _{t}^{\eta}\\
&  =\left(  \theta_{s-2}\right)  _{A}\lambda_{i}^{2}\left(  k\right)
E_{~\eta}^{A}\left(  \delta\phi_{\lambda_{i}}\right)  _{t}^{\eta}\\
&  =\cdots\\
&  =\left(  -1\right)  ^{s}\lambda_{i}^{s}\left(  k\right)  \left(
\theta_{0}\right)  _{A}E_{~\eta}^{A}\left(  \delta\phi_{\lambda_{i}}\right)
_{t}^{\eta}\\
&  =0.
\end{align*}
we proof the assertion.

\bibliographystyle{unsrt}
\bibliography{Strongly_hyp}

\begin{thebibliography}{10}

\bibitem{geroch1996partial}
Robert Geroch.
\newblock Partial differential equations of physics.
\newblock {\em General Relativity, Aberdeen, Scotland}, pages 19--60, 1996.

\bibitem{0264-9381-18-17-202}
Luis Lehner.
\newblock Numerical relativity: a review.
\newblock {\em Classical and Quantum Gravity}, 18(17):R25, 2001.

\bibitem{taylor1991pseudodifferential}
Michael~E Taylor.
\newblock Pseudodifferential operators and nonlinear pde, vol. 100 of progress
  in mathematics, 1991.

\bibitem{taylor1996pseudodifferential}
Michael~E Taylor.
\newblock Pseudodifferential operators.
\newblock In {\em Partial Differential Equations II}, pages 1--73. Springer,
  1996.

\bibitem{nagy2004strongly}
Gabriel Nagy, Omar~E Ortiz, and Oscar~A Reula.
\newblock Strongly hyperbolic second order einstein’s evolution equations.
\newblock {\em Physical Review D}, 70(4):044012, 2004.

\bibitem{schulze1991pseudo}
B-W Schulze.
\newblock {\em Pseudo-differential operators on manifolds with singularities},
  volume~24.
\newblock Elsevier, 1991.

\bibitem{taylor2013partial}
Michael Taylor.
\newblock {\em Partial differential equations II: Qualitative studies of linear
  equations}, volume 116.
\newblock Springer Science \& Business Media, 2013.

\bibitem{calderon1962existence}
AP~Calderon.
\newblock Existence and uniqueness theorems for systems of partial differential
  equations.
\newblock In {\em Proc. Symp. Fluid Dynamics and Appl. Math., Gordon and
  Breach, New York}, pages 147--195, 1962.

\bibitem{hormander1966pseudo}
Lars Hormander.
\newblock Pseudo-differential operators and non-elliptic boundary problems.
\newblock {\em Annals of Mathematics}, pages 129--209, 1966.

\bibitem{hadamard2014lectures}
Jacques Hadamard.
\newblock {\em Lectures on Cauchy's problem in linear partial differential
  equations}.
\newblock Courier Corporation, 2014.

\bibitem{petersen1983introduction}
Bent~E Petersen.
\newblock {\em Introduction to the Fourier transform \& pseudo-differential
  operators}.
\newblock Pitman Advanced Publishing Program, 1983.

\bibitem{nirenberg1973lectures}
Louis Nirenberg.
\newblock {\em Lectures on linear partial differential equations}, volume~17.
\newblock American Mathematical Soc., 1973.

\bibitem{friedrichs1970pseudo}
Kurt~Otto Friedrichs.
\newblock {\em Pseudo-differential operators: an introduction. Lectures given
  in 1967-68 at the Courant Institute}.
\newblock Courant Institute of mathematical sciences, 1970.

\bibitem{treves1980introduction}
Fran{\c{c}}ois Treves.
\newblock {\em Introduction to pseudodifferential and Fourier integral
  operators Volume 2: Fourier integral operators}, volume~2.
\newblock Springer Science \& Business Media, 1980.

\bibitem{lax1963l_2}
Peter~D Lax.
\newblock The l\_2 operator calculus of mikhlin, calderon and zygmund.
\newblock {\em Mimeographed Lecture Note}, 1963.

\bibitem{kohn1973pseudo}
JJ~Kohn.
\newblock Pseudo-differential operators and.
\newblock {\em Partial differential equations}, 23:61, 1973.

\bibitem{hormander1965pseudo}
Lars Hormander.
\newblock Pseudo-differential operators.
\newblock {\em Comm. Pure Appl. Math.}, 18:501--517, 1965.

\bibitem{kreiss1962stabilitatsdefinition}
Heinz-Otto Kreiss.
\newblock {\"U}ber die stabilit{\"a}tsdefinition f{\"u}r differenzengleichungen
  die partielle differentialgleichungen approximieren.
\newblock {\em BIT Numerical Mathematics}, 2(3):153--181, 1962.

\bibitem{gustafsson1995time}
Bertil Gustafsson, Heinz-Otto Kreiss, and Joseph Oliger.
\newblock {\em Time dependent problems and difference methods}, volume~24.
\newblock John Wiley \& Sons, 1995.

\bibitem{kreiss2004initial}
Heinz-Otto Kreiss and Jens Lorenz.
\newblock {\em Initial-boundary value problems and the Navier-Stokes
  equations}.
\newblock SIAM, 2004.

\bibitem{sarbach2012continuum}
Olivier Sarbach and Manuel Tiglio.
\newblock Continuum and discrete initial-boundary value problems and
  einstein’s field equations.
\newblock {\em Living reviews in relativity}, 15(1):9, 2012.

\bibitem{metivier20142}
Guy M{\'e}tivier.
\newblock L 2 well-posed cauchy problems and symmetrizability of first order
  systems [probl{\`e}mes de cauchy bien pos{\'e}s dans l 2 et
  sym{\'e}trisabilit{\'e} pour les syst{\`e}mes du premier ordre].
\newblock {\em Journal de l'{\'E}cole polytechnique-Math{\'e}matiques},
  1:39--70, 2014.

\bibitem{abalos2017necessary}
Fernando Abalos.
\newblock A necessary condition ensuring the strong hyperbolicity of
  first-order systems.
\newblock {\em Journal of Hyperbolic Differential Equations}, 16(01):193--221,
  2019.

\bibitem{Gundlach_2006}
Carsten Gundlach and Jos{\'{e}}~M Mart{\'{\i}}n-Garc{\'{\i}}a.
\newblock Hyperbolicity of second order in space systems of evolution
  equations.
\newblock {\em Classical and Quantum Gravity}, 23(16):S387--S404, jul 2006.

\bibitem{PhysRevD.93.104026}
Pavel Motloch, Wayne Hu, and Hayato Motohashi.
\newblock Self-accelerating massive gravity: Hidden constraints and
  characteristics.
\newblock {\em Phys. Rev. D}, 93:104026, May 2016.

\bibitem{taslaman2014principal}
Leo Taslaman.
\newblock The principal angles and the gap.
\newblock 2014.

\bibitem{afriat1957orthogonal}
Sydney~N Afriat.
\newblock Orthogonal and oblique projectors and the characteristics of pairs of
  vector spaces.
\newblock In {\em Mathematical Proceedings of the Cambridge Philosophical
  Society}, volume~53, pages 800--816. Cambridge University Press, 1957.

\bibitem{strang1967strong}
Gilbert Strang et~al.
\newblock On strong hyperbolicity.
\newblock {\em Journal of Mathematics of Kyoto University}, 6(3):397--417,
  1967.

\bibitem{gantmacher1992theory}
FR~Gantmacher.
\newblock The theory of matrices, vol. 1, chelsea, new york, 1959.
\newblock {\em Google Scholar}, 1992.

\bibitem{gantmakher1998theory}
Feliks~Ruvimovich Gantmakher.
\newblock {\em The theory of matrices}, volume 131.
\newblock American Mathematical Soc., 1998.

\bibitem{moro2002first}
Julio Moro and Froil{\'a}n~M Dopico.
\newblock First order eigenvalue perturbation theory and the newton diagram.
\newblock In {\em Applied Mathematics and Scientific Computing}, pages
  143--175. Springer, 2002.

\bibitem{soderstrom1999perturbation}
Torsten S{\"o}derstr{\"o}m.
\newblock {\em Perturbation results for singular values}.
\newblock Institutionen f{\"o}r informationsteknologi, Uppsala universitet,
  1999.

\bibitem{reula2004strongly}
Oscar~A Reula.
\newblock Strongly hyperbolic systems in general relativity.
\newblock {\em Journal of Hyperbolic Differential Equations}, 1(02):251--269,
  2004.

\bibitem{dedner2002hyperbolic}
Andreas Dedner, Friedemann Kemm, Dietmar Kr{\"o}ner, C-D Munz, Thomas
  Schnitzer, and Matthias Wesenberg.
\newblock Hyperbolic divergence cleaning for the mhd equations.
\newblock {\em Journal of Computational Physics}, 175(2):645--673, 2002.

\bibitem{munz2000three}
C-D Munz, P~Ommes, and R~Schneider.
\newblock A three-dimensional finite-volume solver for the maxwell equations
  with divergence cleaning on unstructured meshes.
\newblock {\em Computer Physics Communications}, 130(1-2):83--117, 2000.

\bibitem{brodbeck1999einstein}
Othmar Brodbeck, Simonetta Frittelli, Peter H{\"u}bner, and Oscar~A Reula.
\newblock Einstein’s equations with asymptotically stable constraint
  propagation.
\newblock {\em Journal of Mathematical Physics}, 40(2):909--923, 1999.

\bibitem{alic2012conformal}
Daniela Alic, Carles Bona-Casas, Carles Bona, Luciano Rezzolla, and Carlos
  Palenzuela.
\newblock Conformal and covariant formulation of the z4 system with
  constraint-violation damping.
\newblock {\em Physical Review D}, 85(6):064040, 2012.

\bibitem{gundlach2005constraint}
Carsten Gundlach, Gioel Calabrese, Ian Hinder, and Jos{\'e}~M
  Mart{\'\i}n-Garc{\'\i}a.
\newblock Constraint damping in the z4 formulation and harmonic gauge.
\newblock {\em Classical and Quantum Gravity}, 22(17):3767, 2005.

\bibitem{rauch2005precise}
Jeffrey Rauch et~al.
\newblock Precise finite speed with bare hands.
\newblock {\em Methods and Applications of Analysis}, 12(3):267--278, 2005.

\end{thebibliography}

\end{document}